\documentclass[runningheads]{llncs}
\usepackage{hyperref}
\usepackage{breakurl}
\usepackage{version}

\usepackage{graphicx}
\usepackage{pstricks}
\usepackage{amsmath}
\usepackage{amssymb}
\usepackage{mathrsfs}

\includeversion{TR}
\excludeversion{CONCUR}

\def\text#1{\textrm{#1}}

\def\precond#1{{\vphantom{#1}}^\bullet #1}
\def\postcond#1{{#1}^\bullet}

\def\production#1{\stackrel{#1}{\longrightarrow}}

\newfont{\fsc}{eusm10 scaled 1100}      % frenchscript letters
\def\powermultiset#1{\bbbn^{#1}}
\def\implies{\Rightarrow}

\def\mathrlap{\mathpalette\mathrlapinternal}
\def\mathrlapinternal#1#2{%
  \rlap{$\mathsurround=0pt#1{#2}$}}
\def\mathllap{\mathpalette\mathllapinternal}
\def\mathllapinternal#1#2{%
  \llap{$\mathsurround=0pt#1{#2}$}}
\def\trail#1{\text{~#1}}

\def\defitem#1{\emph{#1}}

\def\AA{\text{\it AA}}

\let\origexists\exists
\let\orignexists\nexists
\let\origforall\forall
\def\quantorSpace{}
\def\exists#1.{\quantorSpace\origexists\def\quantorSpace{\,}#1.\onespace}
\def\nexists#1.{\quantorSpace\orignexists\def\quantorSpace{\,}#1.\onespace}
\def\forall#1.{\quantorSpace\origforall\def\quantorSpace{\,}#1.\onespace}
\def\onespace#1{\let\argument=#1\ifx\onespace#1\else~\fi\argument}

\let\origmin\min
\def\min{\mathord{\origmin}}
\let\origmax\max
\def\max{\mathord{\origmax}}

\def\quireunderscore{_}
\def\quire#1{%
  \def\tmp{#1}%
  \ifx\tmp\quireunderscore%
    \def\tmp{\quireindexed_}
  \else%
    \def\tmp{\mathscr{Q}#1}
  \fi\tmp}
\def\quireindexed_#1{\mathscr{Q}_{\text{#1}}}

\newtheorem{observation}{Observation}
\def\observationname{Observation}

\newcommand{\refdf}[1]{\definitionname~\ref{df-#1}}

\newcommand{\refthm}[1]{\theoremname~\ref{thm-#1}}

\newcommand{\reflem}[1]{\lemmaname~\ref{lem-#1}}
\newcommand{\reffig}[1]{\figurename~\ref{fig-#1}}

\newcommand{\refobs}[1]{\observationname~\ref{obs-#1}}
\newcommand{\refsec}[1]{Section~\ref{sec-#1}}

\makeatletter
\def\goesto{\@transition\rightarrowfill}
\def\Goesto{\@transition\Rightarrowfill}
\def\ngoesto{\@transition\nrightarrowfill}
\def\nGoesto{\@transition\nRightarrowfill}
\def\@transition#1{\@ifnextchar[{\@@transition{#1}}{\@@transition{#1}[]}}
\newbox\@transbox
\newbox\@arrowbox
\def\rightarrowfill{$\m@th\mathord-\mkern-6mu%
  \cleaders\hbox{$\mkern-2mu\mathord-\mkern-2mu$}\hfill
  \mkern-6mu\mathord\rightarrow$}
\def\Rightarrowfill{$\m@th\mathord=\mkern-6mu%
  \cleaders\hbox{$\mkern-2mu\mathord=\mkern-2mu$}\hfill
  \mkern-6mu\mathord\Rightarrow$}
\def\@@transition#1[#2]%
   {\setbox\@transbox\hbox
       {\vrule height 1.5ex depth 1ex width 0ex\hskip0.25em$\scriptstyle#2$\hskip0.25em}
   \ifdim\wd\@transbox<1.5em
      \setbox\@transbox\hbox to 1.5em{\hfil\box\@transbox\hfil}\fi
   \setbox\@arrowbox\hbox to \wd\@transbox{#1}
   \ht\@arrowbox\z@\dp\@arrowbox\z@
   \setbox\@transbox\hbox{$\mathop{\box\@arrowbox}\limits^{\box\@transbox}$}
   \ht\@transbox\z@\dp\@transbox\z@
   \mathrel{\box\@transbox}}
\makeatother

\def\alignedcaption[#1&#2]{\mbox{\scriptsize $\mathllap{#1{}}\mathrlap{#2}$}}

\def\ie{i.e.\ }

\def\varnothing{\emptyset}

\def\restrictedto{\mathop\upharpoonright}

\def\lastandname{\ and}
\newcommand{\plat}[1]{\raisebox{0pt}[0pt][0pt]{#1}}   % no vertical space
\newcommand{\inp}{\mathbin\in}                        % less space around \in
\def\idx#1#2#3#4#5{
  \def\argone{#1}
  \def\argtwo{#2}
  \def\argthree{#3}
  \def\argfour{#4}
  \def\argfive{#5}
  \def\testprime{'}
  \def\testdprime{''}
  \def\testtprime{'''}
  {\vphantom{\argthree}}_%
    {\vphantom{\argfour}\argone}^%
    {\vphantom{\argfive}{\argtwo}}%
  \argthree_%
    {\vphantom{\argone}\argfour}%
    \ifx\argfive\testprime\argfive\else%
    \ifx\argfive\testdprime\argfive\else%
    \ifx\argfive\testtprime\argfive\else%
    ^{\vphantom{\argtwo}\argfive}\fi\fi\fi%
}

\newcommand{\monus}{\mathrel{\raisebox{-0pt}[0pt][0pt]{$
                      \stackrel{\raisebox{-5pt}[0pt][0pt]{\huge$\cdot$}}
                               {\raisebox{0pt}[0pt][0pt]{$-$}}$}}}

\def\FS{{\it FS}}

\def\swap{\mbox{swap}}
\def\swapeq{\approx_{\mbox{s}}}
\def\adjacent{\mathrel{\,\,\rput(0,0.1){\psscalebox{-0.8}{\leftrightarrow}}\,\,}}
\def\connectedto{\adjacent^*}

\def\FS{\text{FS}}
\def\boldbracketa{%
  \psline[linewidth=0.1pt]{-}(0.15,-0.05)(0.05,-0.05)%
  \psline[linewidth=1.5pt]{-}(0.05,-0.05)(0.05,0.25)%
  \psline[linewidth=0.1pt]{-}(0.05,0.25)(0.15,0.25)%
  \phantom{[\,}%
}
\def\boldbracketb{%
  \psline[linewidth=0.1pt]{-}(0.0,-0.05)(0.1,-0.05)%
  \psline[linewidth=1.5pt]{-}(0.1,-0.05)(0.1,0.25)%
  \psline[linewidth=0.1pt]{-}(0.1,0.25)(0.0,0.25)%
  \phantom{]\,}%
}
\def\BD#1{\mathop{\boldbracketa#1\boldbracketb}}
\def\BDinf#1{{\boldbracketa#1\boldbracketb_{\psscalebox{0.8}{\scriptscriptstyle\infty}}}}
\def\BDify#1{{\it BD}(#1)}
\newenvironment{itemise}{\begin{list}{$\bullet$}{\leftmargin 12pt \labelwidth\leftmargin\advance\labelwidth-\labelsep \topsep 4pt \itemsep 2pt \parsep 2pt}}{\end{list}}
\newenvironment{itemisei}{\begin{list}{$-$}{\leftmargin 12pt \labelwidth\leftmargin\advance\labelwidth-\labelsep \topsep 4pt \itemsep 2pt \parsep 2pt}}{\end{list}}
\def\justempty{}
\newenvironment{define}[1]{\begin{definition}\rm\def\arga{#1}\ifx\justempty\arga~\else#1\fi\\\vspace{-6ex}\\\mbox{~}\begin{itemise}}{\end{itemise}\end{definition}}

\def\lastandname{~{\small\&}}           % between last two authors    
\DeclareFontFamily{T1}{la}{}
\DeclareFontShape{T1}{la}{m}{n}{<->s*[0.8571428571]la14}{}

\def\processfont#1{\text{\fsc #1}}
\def\NN{\processfont{N}}
\def\SS{\processfont{S}\,}
\def\TT{\processfont{T}}
\def\FF{\processfont{F}}
\def\MM{\processfont{M}}
\def\PP{P}
\def\QQ{Q}
\newcommand{\R}{\mathcal{R}}

\usepackage{pstricks}
\usepackage{pst-node}
\usepackage{graphicx}

\newcounter{netimage}
\def\p#1:#2;{\cnode #1{0.3}{n\thenetimage-#2}}
\def\P#1:#2;{\p #1:#2;\pscircle*#1{0.1}}
\def\q#1:#2:#3;{\p #1:#2;\rput#1{\rput[l](0.45,0){\large\it #3}}}
\def\Q#1:#2:#3;{\P #1:#2;\rput#1{\rput[l](0.45,0){\large\it #3}}}
\def\qq#1:#2:#3;{\p #1:#2;\rput#1{\rput[t](0,-0.5){\large\it #3}}}
\def\ql#1:#2:#3;{\p #1:#2;\rput#1{\rput[r](-0.45,0){\large\it #3}}}
\def\qt#1:#2:#3;{\p #1:#2;\rput#1{\rput[b](0,0.45){\large\it #3}}}
\def\Qt#1:#2:#3;{\P #1:#2;\rput#1{\rput[b](0,0.45){\large\it #3}}}
\def\Ql#1:#2:#3;{\P #1:#2;\rput#1{\rput[r](-0.45,0){\large\it #3}}}
\def\qx#1:#2:#3:#4;{\p #1:#2;\rput#1{\rput#4{\large\it #3}}}
\def\QXX#1:#2:#3:#4:#5;{\p #1:#2;\rput#1{\rput#4{\large\it #3}}\pscircle*#5{0.1}}
\def\s#1:#2:#3;{\p #1:#2;\rput#1{\rput(-0.03,0){\large\it #3}}}
\def\t#1:#2:#3;{\rput#1{\rnode{n\thenetimage-#2}{\psframebox{%
  \vbox to 0.6cm{\vfil\hbox to 0.6cm{\hfil\Large\it #3\hfil}\vfil}}}}}
\def\u#1:#2:#3:#4;{\rput#1{\rnode{n\thenetimage-#2}{\psframebox{%
  \vbox to 0.6cm{\vfil\hbox to 0.6cm{\hfil\Large\it #3\hfil}\vfil}}}}%
  \rput#1{\rput[l](0.6,0){\large\it #4}}}
\def\ut#1:#2:#3:#4;{\rput#1{\rnode{n\thenetimage-#2}{\psframebox{%
  \vbox to 0.6cm{\vfil\hbox to 0.6cm{\hfil\Large\it #3\hfil}\vfil}}}}%
  \rput#1{\rput[b](0,0.6){\large\it #4}}}
\def\ul#1:#2:#3:#4;{\rput#1{\rnode{n\thenetimage-#2}{\psframebox{%
  \vbox to 0.6cm{\vfil\hbox to 0.6cm{\hfil\Large\it #3\hfil}\vfil}}}}%
  \rput#1{\rput[r](-0.6,0){\large\it #4}}}
\def\a#1->#2;{\ncline{->}{n\thenetimage-#1}{n\thenetimage-#2}}
\def\A#1->#2;{\ncarc[arcangle=16]{->}{n\thenetimage-#1}{n\thenetimage-#2}}
\def\AA#1->#2;{\ncarc[arcangle=32]{->}{n\thenetimage-#1}{n\thenetimage-#2}}
\def\B#1->#2;{\ncarc[arcangle=-8]{->}{n\thenetimage-#1}{n\thenetimage-#2}}
\def\avlinearc{0.2}
\def\av#1[#2]-#3->[#4]#5;{
  \SpecialCoor
  \psline[linearc=\avlinearc]{->}([angle=#2]n\thenetimage-#1)#3([angle=#4]n\thenetimage-#5)
}

\makeatletter
\long\def\petrinet(#1)#2\end{\psscalebox{0.7}{\pspicture(#1)\stepcounter{netimage}#2\endpspicture}\end}

\makeatother

\psset{arrowsize=5pt}

\titlerunning{On Causal Semantics of Petri Nets}
\begin{CONCUR}
\title{On Causal Semantics of Petri Nets
\texorpdfstring{\newline\large(extended abstract)\thanks{This work was
    partially supported by the DFG (German Research Foundation).}}{}}
\end{CONCUR}
\begin{TR}
\title{On Causal Semantics of Petri Nets%
\texorpdfstring{\thanks{This work was
    partially supported by the DFG (German Research Foundation).}}{}}
\end{TR}
\authorrunning{R.J. van Glabbeek, U. Goltz and J.-W. Schicke}

\author{
  Rob van Glabbeek \inst{1,2} \and
  Ursula Goltz \inst{3} \and
  Jens-Wolfhard Schicke \inst{3}
}
\institute{
  NICTA, Sydney, Australia
  \and
  School of Comp.\ Sc.\ and Engineering, Univ.\ of New South Wales,
  Sydney, Australia
  \and
  Institute for Programming and Reactive Systems, TU Braunschweig, Germany\\[1ex]
   \email{rvg@cs.stanford.edu} \qquad
    \email{goltz@ips.cs.tu-bs.de} \qquad \email{drahflow@gmx.de}
}

\begin{document}

\thispagestyle{empty}
\rput(5.5,-9){\includegraphics[width=14cm]{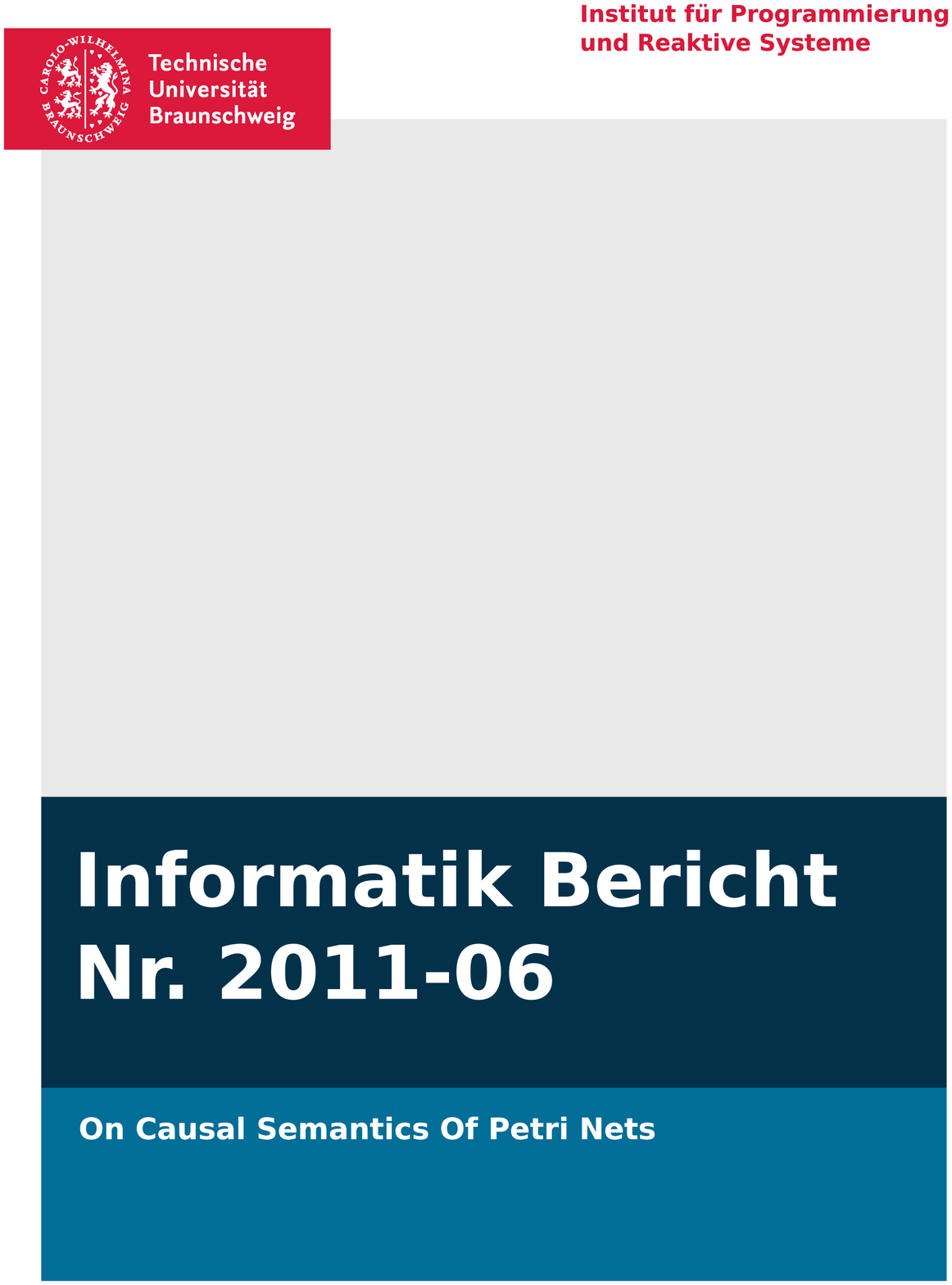}}
\clearpage

\setcounter{page}{1}

\maketitle
\setcounter{footnote}{0}

\begin{abstract}
We consider approaches for causal semantics of Petri nets, explicitly
representing dependencies between transition occurrences.  For one-safe nets
or condition/event-systems, the notion of process as defined by Carl Adam
Petri provides a notion of a run of a system where causal dependencies are
reflected in terms of a partial order.  A well-known problem is how to
generalise this notion for nets where places may carry several tokens.
Goltz and Reisig have defined such a generalisation by distinguishing tokens
according to their causal history.  However, this so-called
\defitem{individual token interpretation} is often considered too detailed.
A number of approaches have tackled the problem of defining a more abstract
notion of process, thereby obtaining a so-called \defitem{collective token
interpretation}.  Here we give a short overview on these attempts and then
identify a subclass of Petri nets, called \defitem{structural conflict nets},
where the interplay between conflict and concurrency due to token multiplicity
does not occur.  For this subclass, we define abstract processes as
equivalence classes of Goltz-Reisig processes.  We justify this approach by
showing that we obtain exactly one maximal abstract process if and only if the
underlying net is conflict-free with respect to a canonical notion of conflict.
\end{abstract}

\section{Introduction}\label{sec-intro}
\noindent
In this paper we address a well-known problem in Petri net theory,
namely how to generalise Petri's concept of non-sequential processes
to nets where places may carry multiple tokens.

One of the most interesting features of Petri nets is that they allow
the explicit representation of causal dependencies between action
occurrences when modelling reactive systems.  This is a key difference
with models of reactive systems (like standard transition systems)
with an inherent so-called interleaving semantics, modelling
concurrency by non-deterministic choice between sequential executions.
In \cite{vanglabbeek01refinement} it has been shown, using the model
of event structures or configuration structures, that causal semantics
are superior to interleaving semantics when giving up the assumption
that actions are atomic entities.

In the following, we give a concise overview on existing approaches on
semantics of Petri nets that give an account of their runs, without
claiming completeness, and following closely a similar presentation in
\cite{glabbeek11ipl}.
\advance\textheight 1pt

Initially, Petri introduced the concept of a net together with a definition of
its dynamic behaviour in terms of the firing rule for single transitions or
for finite sets (\defitem{steps}) of transitions firing in parallel. Sequences
of transition firings or of steps are the usual way to define the behaviour of
a Petri net. When considering only single transition firings, the set of all
firing sequences yields a linear time interleaving semantics (no choices
between alternative behaviours are represented). Otherwise we obtain a linear
time step semantics, with information on possible parallelism, but without
explicit representation of causal dependencies between transition occurrences.

\advance\textheight -1pt
Petri then defined \defitem{condition/event systems}, where ---
amongst other restrictions --- places (there called conditions) may
carry at most one token. For this class of nets, he proposed what is
now the classical notion of a \defitem{process}, given as a mapping
from an \defitem{occurrence net} (acyclic net with unbranched places)
to the original net \cite{petri77nonsequential,genrich80dictionary}.
A process models a run of the represented system, obtained by choosing
one of the alternatives in case of conflict. It records all
occurrences of the transitions and places visited during such a run,
together with the causal dependencies between them, which are given by
the flow relation of the net.  A linear-time causal semantics of a
condition/event system is thus obtained by associating with a net the
set of its processes.  Depending on the desired level of abstraction,
it may suffice to extract from each process just the partial order of
transition occurrences in it. The firing sequences of transitions or
steps can in turn be extracted from these partial orders.  Nielsen,
Plotkin and Winskel extended this to a branching-time semantics by
using occurrence nets with forward branched places
\cite{nielsen81petri}.  These capture all runs of the represented
system, together with the branching structure of choices between them.

However, the most frequently used class of Petri nets are nets where places
may carry arbitrary many tokens, or a certain maximal number of tokens when
adding place capacities. This type of nets is often called
\defitem{place/transition systems} (\mbox{P\hspace{-1pt}/T} systems). Here
tokens are usually assumed to be indistinguishable entities, for example
representing a number of available resources in a system. Unfortunately, it is
not straightforward to generalise the notion of process, as defined by Petri
for condition/event systems, to \mbox{P\hspace{-1pt}/T} systems. In fact, it
has now for more than 20 years been a well-known problem in Petri net theory
how to formalise an appropriate causality-based concept of process or run for
general \mbox{P\hspace{-1pt}/T} systems. In the following we give an
introduction to the problem and a short overview on existing approaches.

As a first approach, Goltz and Reisig generalised Petri's notion of process to
general \mbox{P\hspace{-1pt}/T} systems \cite{goltz83nonsequential}. We call
this notion of a process \defitem{GR-process}. It is based on a
canonical unfolding of a P/T systems into a condition/event system,
representing places that may carry several tokens by a corresponding number
of conditions (see \cite{goltz87representations}).  \reffig{unsafe} shows a
\mbox{P\hspace{-1pt}/T} system with two of its GR-processes.

\begin{figure}
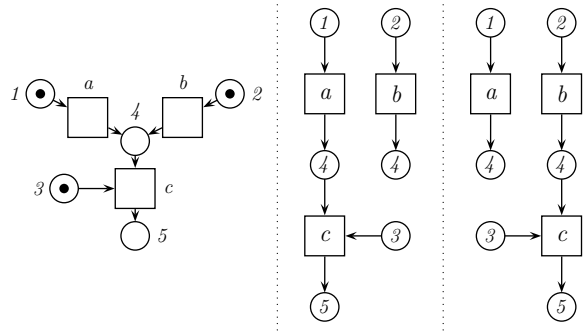

  \begin{center}
    \psscalebox{0.9}{
    \begin{petrinet}(12,7)
      \Ql(0.5,5):p:1;
      \Q(4.5,5):q:2;
      \ut(1.5,4.5):a::a;
      \ut(3.5,4.5):b::b;
      \qt(2.5,4):r:4;
      \Ql(1,3):cs:3;
      \u(2.5,3):c::c;
      \q(2.5,2):s:5;

      \a p->a; \a q->b; \a a->r; \a b->r;
      \a r->c; \a cs->c; \a c->s;

      \psline[linestyle=dotted](5.5,0)(5.5,7)

      \s(6.5,6.5):p2:1;
      \s(8.0,6.5):q2:2;
      \t(6.5,5):a2:a;
      \t(8.0,5):b2:b;
      \s(6.5,3.5):r2:4;
      \s(8.0,3.5):r2p:4;
      \t(6.5,2):c2:c;
      \s(8.0,2):cs2:3;
      \s(6.5,0.5):s2:5;

      \a p2->a2; \a q2->b2; \a a2->r2; \a b2->r2p; \a r2->c2; \a cs2->c2; \a c2->s2;

      \psline[linestyle=dotted](9.0,0)(9.0,7)

      \s(10.0,6.5):p3:1;
      \s(11.5,6.5):q3:2;
      \t(10.0,5):a3:a;
      \t(11.5,5):b3:b;
      \s(10.0,3.5):r3p:4;
      \s(11.5,3.5):r3:4;
      \t(11.5,2):c3:c;
      \s(10.0,2):cs3:3;
      \s(11.5,0.5):s3:5;

      \a p3->a3; \a q3->b3; \a a3->r3p; \a b3->r3; \a r3->c3; \a cs3->c3; \a c3->s3;

    \end{petrinet}
    }
  \end{center}
  \vspace{-2ex}
  \caption{A net $N$ with its two maximal GR-processes. The correspondence between elements of the net and their occurrences in the processes is indicated by labels.}
  \label{fig-unsafe}
\end{figure}

Engelfriet adapted GR-processes by additionally representing choices
between alternative behaviours \cite{engelfriet91branchingprocesses},
thereby adopting the approach of \cite{nielsen81petri} to
\mbox{P\hspace{-1pt}/T} systems, although without arc weights.
Meseguer, Sassone and Montanari extended this to cover also arc weights \cite{MMS97}.

However, if one wishes to interpret \mbox{P\hspace{-1pt}/T} systems with a
causal semantics, there are alternative interpretations of what ``causal
semantics'' should actually mean. Goltz already argued that when abstracting
from the identity of multiple tokens residing in the same place, GR-processes
do not accurately reflect runs of nets, because if a Petri net is
conflict-free it should intuitively have only one complete run (for there are
no choices to resolve), yet it may have multiple maximal GR-processes
\cite{goltz86howmany}.  This phenomenon already occurs in \reffig{unsafe},
since the choice between alternative behaviours is here only due to the
possibility to choose between two tokens which can or even should be seen as
indistinguishable entities.  A similar argument is made, e.g., in
\cite{HKT95}.

At the heart of this issue is the question whether multiple tokens residing in
the same place should be seen as individual entities, so that a transition
consuming just one of them constitutes a conflict, as in the interpretation
underlying GR-processes and the approach of
\cite{engelfriet91branchingprocesses,MMS97}, or whether such tokens are
indistinguishable, so that taking one is equivalent to taking the other.  Van
Glabbeek and Plotkin call the former viewpoint the \defitem{individual token
interpretation} of P\hspace{-1pt}/T systems.  For an alternative
interpretation, they use the term \defitem{collective token interpretation}
\cite{glabbeek95configuration}.  A possible formalisation of these
interpretations occurs in \cite{glabbeek05individual}. In the following we
call process notions for \mbox{P\hspace{-1pt}/T~systems} which are adherent to
a collective token philosophy \defitem{abstract processes}.  Another option,
proposed by Vogler, regards tokens only as notation for a natural number
stored in each place; these numbers are incremented or decremented when firing
transitions, thereby introducing explicit causality between any transitions
removing tokens from the same place \cite{vogler91executions}.

Mazurkiewicz applies again a different approach in \cite{mazurkiewicz89multitree}.
He proposes \defitem{multitrees}, which record possible multisets of fired
transitions, and then takes confluent subsets of multitrees as abstract
processes of P\hspace{-1pt}/T systems. This approach does not explicitly
represent dependencies between transition occurrences and hence does not apply
to nets with self-loops, where such information may not always be retrieved.

Yet another approach has been proposed by Best and Devillers in
\cite{best87both}.  Here an equivalence relation is generated by a
transformation for changing causalities in GR-processes, called
\defitem{swapping}, that identifies GR-processes which differ only in
the choice which token was removed from a place. In this paper, we
adopt this approach and we show that it yields a fully satisfying
solution for a subclass of P\hspace{-1pt}/T systems. We call the
resulting notion of a more abstract process \defitem{BD-process}. In
the special case of one-safe \mbox{P\hspace{-1pt}/T} systems (where
places carry at most one token), or for condition/event systems, no
swapping is possible, and a BD-process is just an isomorphism class of
GR-processes.

Meseguer and Montanari formalise runs in a net $N$ as
morphisms in a category $\mathcal{T}(N)$ \cite{MM88}. In \cite{DMM89}
it has been established that these morphisms ``coincide with the
commutative processes defined by Best and Devillers'' (their
terminology for BD-processes). Likewise, Hoogers, Kleijn and
Thiagarajan represent an abstract run of a net by a \defitem{trace},
thereby generalising the trace theory of Mazurkiewicz
\cite{mazurkiewicz95tracetheory}, and remark that ``it is
straightforward but laborious to set up a 1-1 correspondence between
our traces and the equivalence classes of finite processes generated
by the swap operation in [Best and Devillers, 1987]''.

To explain why it can be argued that BD-processes are not fully satisfying as
abstract processes for general \mbox{P\hspace{-1pt}/T} systems, we recall in
\reffig{badswapping} an example due to Ochma\'nski
\cite{ochmanski89personal,barylska09nonviolence}, see also
\cite{DMM89,glabbeek11ipl}. In the initial situation only two of the three
enabled transitions can fire, which constitutes a conflict.  However, the
equivalence obtained from the swapping transformation (formally defined in
\refsec{semantics}) identifies all possible maximal GR-processes and hence
yields only one complete abstract run of the system.  We are not aware of a
solution, i.e.\ any formalisation of the concept of a run of a net that
correctly represents both causality and parallelism of nets, and meets the
requirement that for this net there is more than one possible complete run.

\begin{figure}[t]
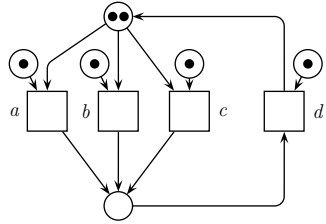

\vspace*{1em}
  \begin{center}
    \psscalebox{0.9}{
    \begin{petrinet}(12,5)
      \P(3.5,4):pa;
      \P(5,4):pb;
      \P(7,4):pc;
      \P(9.5,4):pd;

      \ul(4,3):a::a;
      \ul(5.5,3):b::b;
      \u(7,3):c::c;
      \u(9,3):d::d;

      \p(5.5,5):p;
      \p(5.5,1):q;

      \av p[210]-(4,4)->[90]a; \a pa->a; \a a->q;
      \a p->b; \a pb->b; \a b->q;
      \a p->c; \a pc->c; \a c->q;
      \av q[0]-(9,1)->[270]d; \a pd->d; \av d[90]-(9,5)->[0]p;

      \pscircle*(5.38,5){0.1}
      \pscircle*(5.62,5){0.1}
    \end{petrinet}
    }
  \end{center}
  \vspace{-5ex}
  \caption{A net with only a single process up to swapping equivalence.$\!$}
  \label{fig-badswapping}
\end{figure}

In \cite{glabbeek11ipl} and in the present paper, we continue the line of
research of \cite{MM88,DMM89,mazurkiewicz89multitree,HKT95} to formalise a
causality-based notion of an abstract process of a \mbox{P\hspace{-1pt}/T}
system that fits a collective token interpretation.  As remarked already in
\cite{goltz86howmany}, `what we need is some notion of an ``abstract
process''' and `a notion of maximality for abstract processes', such that `a
\mbox{P\hspace{-1pt}/T}-system is conflict-free iff it has exactly one maximal
abstract process starting at the initial marking'.  The example from
\reffig{badswapping} shows that BD-processes are in general not suited. We
defined in \cite{glabbeek11ipl} a subclass of \mbox{P\hspace{-1pt}/T} systems
where conflict and concurrency are clearly separated. We called these nets
\defitem{structural conflict nets}.  Using the formalisation of conflict for
\mbox{P\hspace{-1pt}/T} systems from \cite{goltz86howmany}, we have shown
that, for this subclass of P/T systems, we obtain more than one maximal
BD-process whenever the net contains a conflict.\footnote{The notion of
maximality for BD-processes is not trivial. However, with the results from
\refsec{results}, Corollary~1 from \cite{glabbeek11ipl} may be rephrased in
this way.} The proof of this result is quite involved; it was achieved by
using an alternative characterisation of BD-processes via firing sequences
from \cite{best87both}.

In this paper, we will show the reverse direction of this result, namely that
we obtain exactly one maximal BD-process of a structural conflict net if the
net is conflict-free.  Depending on the precise formalisation of a suitable
notion of maximality of BD-processes, this holds even for arbitrary nets.
Summarising, we then have established that we obtain exactly one maximal
abstract process in terms of BD-processes for structural conflict nets
\emph{if and only if} the net is conflict-free with respect to a canonical
notion of conflict.

We proceed by defining basic notions for \mbox{P\hspace{-1pt}/T} systems in
Section \ref{sec-basic}. In Section \ref{sec-semantics}, we define
GR-processes and introduce the swapping equivalence. Section
\ref{sec-conflict} recalls the concept of conflict in \mbox{P\hspace{-1pt}/T}
systems and defines structural conflict nets.\footnote{The material in
Sections \ref{sec-basic} to \ref{sec-conflict} follows closely the
presentation in \cite{glabbeek11ipl}, but needs to be included to make the
paper self-contained.} In Section \ref{sec-finiteruns}, we recapitulate the
alternative characterisation of BD-processes from \cite{best87both} in terms
of an equivalence notion on firing sequences \cite{best87both} and prove in
this setting that a conflict-free net has exactly one maximal run.  Finally,
in Section \ref{sec-results}, we investigate notions of maximality for
BD-processes and then transfer the result from Section \ref{sec-finiteruns} to
BD-processes.
\begin{CONCUR}
Due to lack of space, the proofs of Lemma's~\ref{lem-bdprocessexists},
\ref{lem-procleqimplrunleq}, \ref{lem-closingdiamond}, \ref{lem-cfmaximals}
and~\ref{lem-scmaximals}, some quite involved,
are omitted; these can be found in \cite{glabbeek11causaltr}.
\end{CONCUR}

\section{Place/transition Systems}\label{sec-basic}

\noindent
We will employ the following notations for multisets.

\begin{define}{
  Let $X$ be a set.
}\label{df-multiset}
\item A {\em multiset} over $X$ is a function $A\!:X \rightarrow \bbbn$,
i.e.\ $A\in \powermultiset{X}\!\!$.
\item $x \in X$ is an \defitem{element of} $A$, notation $x \in A$, iff $A(x) > 0$.
\item For multisets $A$ and $B$ over $X$ we write $A \subseteq B$ iff
 \mbox{$A(x) \leq B(x)$} for all $x \inp X$;
\\ $A\cup B$ denotes the multiset over $X$ with $(A\cup B)(x):=\text{max}(A(x), B(x))$,
\\ $A\cap B$ denotes the multiset over $X$ with $(A\cap B)(x):=\text{min}(A(x), B(x))$,
\\ $A + B$ denotes the multiset over $X$ with $(A + B)(x):=A(x)+B(x)$,
\\ $A - B$ is given by
$(A - B)(x):=A(x)\monus B(x)=\mbox{max}(A(x)-B(x),0)$, and
for $k\inp\bbbn$ the multiset $k\cdot A$ is given by
$(k \cdot A)(x):=k\cdot A(x)$.
\item The function $\emptyset\!:X\rightarrow\bbbn$, given by
  $\emptyset(x):=0$ for all $x \inp X$, is the \emph{empty} multiset over $X$.
\item If $A$ is a multiset over $X$ and $Y\subseteq X$ then
  $A\restrictedto Y$ denotes the multiset over $Y$ defined by
  $(A\restrictedto Y)(x) := A(x)$ for all $x \inp Y$.
\item The cardinality $|A|$ of a multiset $A$ over $X$ is given by
  $|A| := \sum_{x\in X}A(x)$.
\item A multiset $A$ over $X$ is \emph{finite}
  iff $|A|<\infty$, i.e.,
  iff the set $\{x \mid x \inp A\}$ is finite.
\end{define}
Two multisets $A\!:X \rightarrow \bbbn$ and
$B\!:Y\rightarrow \bbbn$
are \emph{extensionally equivalent} iff
$A\restrictedto (X\cap Y) = B\restrictedto (X\cap Y)$,
$A\restrictedto (X\setminus Y) = \emptyset$, and
$B \restrictedto (Y\setminus X) = \emptyset$.
In this paper we often do not distinguish extensionally equivalent
multisets. This enables us, for instance, to use $A \cup B$ even
when $A$ and $B$ have different underlying domains.
With $\{x,x,y\}$ we will denote the multiset over $\{x,y\}$ with
$A(x)\mathbin=2$ and $A(y)\mathbin=1$, rather than the set $\{x,y\}$ itself.
A multiset $A$ with $A(x) \leq 1$ for all $x$ is
identified with the set $\{x \mid A(x)=1\}$.

Below we define place/transition systems as net structures with an initial marking.
In the literature we find slight variations in the definition of \mbox{P\hspace{-1pt}/T}
systems concerning the requirements for pre- and postsets of places
and transitions. In our case, we do allow isolated places. For
transitions we allow empty postsets, but require at least one
preplace, thus avoiding problems with infinite self-concurrency.
Moreover, following \cite{best87both}, we restrict attention
to nets of \defitem{finite synchronisation}, meaning that each
transition has only finitely many pre- and postplaces.
Arc weights are included by defining the flow relation as a function to the natural numbers.
For succinctness, we will refer to our version of a \mbox{P\hspace{-1pt}/T} system as a \defitem{net}.

\begin{define}{}\label{df-nst}
\item[]
  A \defitem{net} is a tuple
  $N = (S, T, F, M_0)$ where
  \begin{itemise}
    \item $S$ and $T$ are disjoint sets (of \defitem{places} and \defitem{transitions}),
    \item $F: (S \mathord\times T \mathrel\cup T \mathord\times S) \rightarrow \bbbn$
      (the \defitem{flow relation} including \defitem{arc weights}), and
    \item $M_0 : S \rightarrow \bbbn$ (the \defitem{initial marking})
  \end{itemise}
  such that for all $t \inp T$ the set $\{s\mid F(s, t) > 0\}$ is
  finite and non-empty, and the set $\{s\mid F(t, s) > 0\}$ is finite.
\end{define}

\noindent
Graphically, nets are depicted by drawing the places as circles and
the transitions as boxes. For $x,y \inp S\cup T$ there are $F(x,y)$
arrows (\defitem{arcs}) from $x$ to $y$.  When a net represents a
concurrent system, a global state of this system is given as a
\defitem{marking}, a multiset of places, depicted by placing $M(s)$
dots (\defitem{tokens}) in each place $s$.  The initial state is
$M_0$.  The system behaviour is defined by the possible moves between
markings $M$ and $M'$, which take place when a finite multiset $G$ of
transitions \defitem{fires}.  When firing a transition, tokens on
preplaces are consumed and tokens on postplaces are created, one for
every incoming or outgoing arc of $t$, respectively.  Obviously, a
transition can only fire if all necessary tokens are available in $M$
in the first place. \refdf{firing} formalises this notion of behaviour.
\begin{TR}\newpage\end{TR}

\begin{define}{
  Let $N\!=\!(S, T, F, M_0)$ be a net and $x\inp S\cup T$.
}\label{df-preset}
\item[]
The multisets $\precond{x},~\postcond{x}: S\cup T \rightarrow
\bbbn$ are given by $\precond{x}(y)=F(y,x)$ and
$\postcond{x}(y)=F(x,y)$ for all $y \inp S \cup T$.
If $x\in T$, the elements of $\precond{x}$ and $\postcond{x}$ are
called \emph{pre-} and \emph{postplaces} of $x$, respectively.
These functions extend to multisets
$X:S \cup T \rightarrow\bbbn$ as usual, by
$\precond{X} := \Sigma_{x \in S \cup T}X(x)\cdot\precond{x}$ and
$\postcond{X} := \Sigma_{x \in S \cup T}X(x)\cdot\postcond{x}$.
\end{define}

\begin{define}{
  Let $N \mathbin= (S, T, F, M_0)$ be a net,
  $G \in \bbbn^T\!$, $G$ non-empty and finite, and $M, M' \in \bbbn^S\!$.
}\label{df-firing}
\item[]
$G$ is a \defitem{step} from $M$ to $M'$,
written $M\production{G}_N M'$, iff
\begin{itemise}
  \item $^\bullet G \subseteq M$ ($G$ is \defitem{enabled}) and
  \item $M' = (M - \mbox{$^\bullet G$}) + G^\bullet$. 
\end{itemise}
We may leave out the subscript $N$ if clear from context.
Extending the notion to words $\sigma = t_1t_2\ldots t_n \in T^*$
we write $M\production{\sigma} M'$ for\vspace{-5pt}
$$
\exists M_1, M_2, \ldots, M_{n-1}.
M\!\production{\{t_1\}}\! M_1\!\production{\{t_2\}}\! M_2 \cdots M_{n-1}\!\production{\{t_n\}}\! M'\!\!.
$$
When omitting $\sigma$ or $M'$ we always mean it to be existentially quantified.
When $M_0 \production{\sigma}_N$, the word $\sigma$
is called a \defitem{firing sequence} of $N$.
The set of all firing sequences of $N$ is denoted by $\FS(N)$.
\end{define}

\noindent
Note that steps are (finite) multisets, thus allowing self-concurrency.
Also note that $M\goesto[\{t,u\}]$ implies $M\goesto[tu]$ and $M\goesto[ut]$.
We use the notation $t\in \sigma$ to indicate that the transition $t$
occurs in the sequence $\sigma$, and $\sigma\leq\rho$ to indicate that
$\sigma$ is a prefix of the sequence $\rho$, i.e.\ $\exists \mu. \rho=\sigma\mu$.

\section{Processes of place/transition systems}\label{sec-semantics}

\noindent
We now define processes of nets.
A (GR-)process is essentially a conflict-free, acyclic net together
with a mapping function to the original net. It can be obtained by
unwinding the original net, choosing one of the alternatives in case
of conflict.
The acyclic nature of the process gives rise to a notion of causality
for transition firings in the original net via the mapping function.
Conflicts present in the original net are represented by one net yielding
multiple processes, each representing one possible way to decide the conflicts.

\begin{define}{}\label{df-process}
 \item[]
  A pair $\PP = (\NN, \pi)$ is a
  \defitem{(GR-)process} of a net $N = (S, T, F, M_0)$
  iff
  \begin{itemise}\itemsep 3pt
   \item $\NN = (\SS, \TT, \FF, \MM_0)$ is a net, satisfying
   \begin{itemisei}
    \item $\forall s \in \SS. |\precond{s}| \leq\! 1\! \geq |\postcond{s}|
    \wedge\, \MM_0(s) = \left\{\begin{array}{@{}l@{\quad}l@{}}1&\mbox{if $\precond{s}=\emptyset$}\\
                                   0&\mbox{otherwise,}\end{array}\right.$
    \item $\FF$ is acyclic, \ie
      $\forall x \inp \SS \cup \TT. (x, x) \mathbin{\not\in} \FF^+$,
      where $\FF^+$ is the transitive closure of $\{(t,u)\mid F(t,u)>0\}$,
    \item and $\{t \mid (t,u)\in \FF^+\}$ is finite for all $u\in \TT$.
   \end{itemisei}
    \item $\pi:\SS \cup \TT \rightarrow S \cup T$ is a function with 
    $\pi(\SS) \subseteq S$ and $\pi(\TT) \subseteq T$, satisfying
   \begin{itemisei}
    \item $\pi(\MM_0) = M_0$, i.e.\ $M_0(s) = |\pi^{-1}(s) \cap \MM_0|$ for all $s\in S$, and
    \item $\forall t \in \TT, s \in S.
      F(s, \pi(t)) = |\pi^{-1}(s) \cap \precond{t}| \wedge
      F(\pi(t), s) = |\pi^{-1}(s) \cap \postcond{t}|$.
  \end{itemisei}
  \end{itemise}
  $P$ is called \defitem{finite} if $\TT$ is finite.
\end{define}

\noindent
The conditions for $\NN$ ensure that a process is indeed a mapping from an
occurrence net as defined in \cite{petri77nonsequential,genrich80dictionary}
to the net $N$; hence we define processes here in the classical way as in
\cite{goltz83nonsequential,best87both} (even though not introducing occurrence
nets explicitly).

A process is not required to represent a completed run of the original net.
It might just as well stop early. In those cases, some set of transitions can
be added to the process such that another (larger) process is obtained. This
corresponds to the system taking some more steps and gives rise to a natural
order between processes.

\begin{define}{
  Let $\PP = ((\SS, \TT, \FF, \MM_0), \pi)$ and\\ $\PP' = ((\SS', \TT\,', \FF\,', \MM_0'), \pi')$ be
  two processes of the same net.
}\label{df-extension}
\item
  $\PP'$ is a \defitem{prefix} of $\PP$, notation $\PP'\leq \PP$, and 
  $\PP$ an \defitem{extension} of $\PP'$, iff 
    $\SS'\subseteq \SS$,
    $\TT\,'\subseteq \TT$,
    $\MM_0' = \MM_0$,
    $\FF\,'=\FF\restrictedto(\SS' \mathord\times \TT\,' \mathrel\cup \TT\,' \mathord\times \SS')$
    and $\pi'=\pi\restrictedto(\SS'\times \TT\,')$.
\item
  A process of a net is said to be \defitem{maximal} if 
  it has no proper extension.
\end{define}

\noindent
The requirements above imply that if $\PP'\leq \PP$, $(x,y)\in
\FF^+$ and $y\in \SS' \cup \TT\,'$ then $x\in \SS' \cup \TT\,'$.
Conversely, any subset $\TT\,'\subseteq \TT$ satisfying
$(t,u)\in \FF^+ \wedge u\in \TT\,' \Rightarrow t\in \TT\,'$ uniquely determines a
prefix of $\PP$.

\begin{TR}
Henceforth, we will write $P'\goesto[G]P$ with $G\inp\bbbn^T$ a finite and non-empty
multiset of transitions of the underlying net, if $P'\mathbin\leq P$, all transitions in
$K:=\TT\setminus\TT\,'$ are maximal in $\TT$ w.r.t.\ $\FF^+$, and $\pi(K)=G$, i.e.\
$G(t)=|\pi^{-1}(t)\cap K|$ for all $t\in T$. As usual, we write $P'\goesto[a]P$ instead of $P'\goesto[\{a\}]P$
for singleton steps.
Let $\mathcal{P}_0(N)$ be the set of \emph{initial processes} of a net $N$: those with an
empty set of transitions.
Now for each finite process $P$ of $N$, having $n$ transitions, there is a sequence
$P_0 \goesto[a_1] P_1 \goesto[a_2] \ldots \goesto[a_n] P_n$ with
$P_0\in\mathcal{P}_0(N)$ and $P_n=P$.

For $P=((\SS,\TT,\FF,\MM_0),\pi)$ a finite GR-process of a net $N=(S,T,F,M_0)$, we write
$P^\circ$ for $\{s\inp \SS \mid \forall t\inp \TT. \FF(s,t)\mathbin=0\}$, and
$\widehat P$ for the marking $\pi(P^\circ)\in \bbbn^S$.
The following observations describe a \emph{step bisimulation}
\cite{vanglabbeek01refinement} between the above transition relation on the
processes of a net, and the one on its markings.

\begin{observation}\rm\label{obs-bisimulation}
Let $N=(S,T,F,M_0)$ be a net, $G\inp\bbbn^T$ non-empty and finite, and $P,Q$ be
finite GR-processes of $N$.  \vspace{-1ex}
\begin{itemize}
\item[(a)]
$\mathcal{P}_0(N)\neq\emptyset$ and if $P\inp\mathcal{P}_0$ then $\widehat P = M_0$.
\item[(b)]
If $P \goesto[G] Q$ then $\widehat P \goesto[G] \widehat Q$.
\item[(c)]
If $\widehat P \goesto[G] M$ then there is a $Q$ with  $P \goesto[G] Q$ and $\widehat Q = M$.
\item[(d)]
$\widehat P$ is \emph{reachable} in the sense that $M_0 \goesto[] \widehat{P}$.
(This follows from (a) and (b).)
\end{itemize}
\end{observation}
\end{TR}
Two processes $(\NN, \pi)$ and $(\NN\,', \pi')$
are \defitem{isomorphic} iff there exists an isomorphism $\phi$ from
$\NN$ to $\NN\,'$ which respects the process mapping, i.e.\
$\pi = \pi' \circ \phi$.
Here an isomorphism $\phi$ between two nets $\NN=(\SS, \TT, \FF,
\MM_0)$ and $\NN\,'=(\SS', \TT\,', \FF\,', \MM'_0)$ is a
bijection between their places and transitions such that
$\MM'_0(\phi(s))=\MM_0(s)$ for all $s\in\SS$ and
$\FF\,'(\phi(x),\phi(y))=\FF(x,y)$ for all $x,y\in \SS\cup\TT$.
\vspace{2ex}

Next we formally introduce the swapping transformation and the resulting
equivalence notion on GR-processes from \cite{best87both}.

\begin{define}{
  Let $\PP = ((\SS, \TT, \FF, \MM_0), \pi)$ be a process and
  let $p, q \in \SS$ with $(p,q) \notin \FF^+\cup (\FF^+)^{-1}$ and
  $\pi(p) = \pi(q)$.
  }
\label{df-swap}
\item[]
  Then $\swap(\PP, p, q)$ is defined as $((\SS, \TT, \FF\,', \MM_0), \pi)$ with
  \begin{equation*}
    \FF\,'(x, y) = \begin{cases}
      \FF(q, y) & \text{ iff } x = p,\, y \in \TT\\
      \FF(p, y) & \text{ iff } x = q,\, y \in \TT\\
      \FF(x, y) & \text{ otherwise. }
    \end{cases}
  \end{equation*}
\end{define}

\begin{define}{}\label{df-swapeq}
\item
  Two processes $\PP$ and $\QQ$ of the same net are
  \defitem{one step swapping equivalent} ($\PP \swapeq \QQ$) iff
  $\swap(\PP, p, q)$ is isomorphic to $\QQ$ for some places $p$ and $q$.

\item
We write $\swapeq^*$ for the reflexive and transitive closure of $\swapeq$,
and $\BD{\PP}$ for the $\swapeq^*$-equivalence class of a finite process $\PP$.
The prefix relation $\leq$ between processes is lifted to such
equivalence classes by $\BD{\PP'} \leq \BD{\PP}$ iff
$\PP' \swapeq^* \QQ' \leq \QQ \swapeq^* \PP$ for some $\QQ',\QQ$.

\item
  Two processes $\PP$ and $\QQ$ are \defitem{swapping equivalent}
  ($\PP \swapeq^\infty \QQ$) iff
  \begin{equation*}
    \begin{split}
  &{\downarrow(\{\BD{\PP'}\mid \PP'\leq \PP,~\PP'~\mbox{finite}\})}
  =\\&{\downarrow(\{\BD{\QQ'}\mid \QQ'\leq \QQ,~\QQ'~\mbox{finite}\})}
    \end{split}
  \end{equation*}
  where $\downarrow$ denotes prefix-closure under $\leq$.

\item
$\!$We call a $\swapeq^\infty$\hspace{-2pt}-equivalence class of processes
a \defitem{{\small BD-}process}, and write $\BDinf{P}$.
\end{define}
It is not hard to verify that if $\PP \swapeq^* \QQ \leq \QQ'$ then $\PP \leq
\PP' \swapeq^* \QQ'$ for some process $\PP'$. This implies that $\leq$ is a
partial order on $\swapeq^*$-equivalence classes of finite processes.
Alternatively, this conclusion follows from Theorem 4 in \cite{glabbeek11ipl}.

Our definition of $\swapeq^\infty$ deviates from the definition of
$\equiv_1^\infty$ from \cite{best87both} to make proofs easier later
on.  We conjecture however that the two notions coincide.

Note that if $\PP \swapeq^\infty \QQ$ and $\PP$ is finite, then also
$\QQ$ is finite. Moreover, for finite GR-processes $\PP$ and $\QQ$ we
have $\PP \swapeq^\infty \QQ$ iff $\PP \swapeq^* \QQ$.
Thus, for a finite GR-process $\PP$, we have $\BDinf{P}=\BD{P}$.
In that case we call $\BD{P}$ a \emph{finite} BD-process.

\begin{TR}
The following observations are easy to check.
\begin{observation}\rm\label{obs-swaptrans}
Let $P,Q,P',Q'$ be finite GR-processes of a net $N$.  \vspace{-1ex}
\begin{itemize}
\item[(a)] 
If $P\goesto[a]Q$ and $P\goesto[a]Q'$ then $Q\swapeq^* Q'$.
\item[(b)]
If $P\swapeq^*Q$ and $P\goesto[a]P'$ then $Q\goesto[a]Q'$ for some $Q'$ with
$P'\swapeq^* Q'$.
\end{itemize}
\end{observation}
For GR-processes $P$ and $Q$ we write $\BD{P}\goesto[a]\BD{Q}$ if $P\goesto[a]Q'$
for some $Q'\inp\BD{Q}$. By \refobs{swaptrans}(b) this implies that
for any $P'\inp\BD{P}$ there is a $Q'\inp\BD{Q}$ with $P'\goesto[a]Q'$.
By \refobs{swaptrans}(a), for any BD-process $\BD{P}$ of a $N$ and any transition
$a$ of $N$ there is at most one BD-process $\BD{Q}$ with $\BD{P}\goesto[a]\BD{Q}$.
\end{TR}
\vspace{2ex}

We define a \emph{BD-run} as a more abstract and more general form
of BD-process. Like a BD-process, a BD-run is completely determined by
its finite approximations, which are finite BD-processes; however, a
BD-run does not require that these finite approximations are generated
by a given GR-process.

\begin{define}{
  Let $N$ be a net.
}\label{bd-run}
\item[]
  A \defitem{BD-run} $\R$ of $N$ is a non-empty set of finite BD-processes of $N$ such that
  \begin{itemise}
    \item $\BD{P} \leq \BD{Q} \in \R \implies \BD{P} \in \R$ ($\R$ is prefix-closed), and
    \item $\BD{P}, \BD{Q} \in \R \implies \exists \BD{U} \in \R. \BD{P}
      \leq \BD{U} \wedge \BD{Q} \leq \BD{U}$ ($\R$ is directed).
  \end{itemise}
\end{define}
The class of finite BD-processes and the finite elements (in the set
theoretical sense) in the class of BD-runs are in bijective correspondence.
Every finite BD-run $\R$ must have a largest element, say
$\BD{P}$, and the set of all prefixes of $\BD{P}$ is $\R$.
Conversely, the set of prefixes of a finite BD-process $\BD{P}$  
is a finite BD-run of which the largest element is again $\BD{P}$.

We now define a canonical mapping from GR-processes to BD-runs.

\begin{define}{
  Let $N$ be a net and $\PP$ a process thereof.
}\label{def-bdify}
\item[]
  Then $\BDify{\PP} := \mathord\downarrow \{\BD{\PP'} \mid
    \PP' \leq \PP,~ \PP' \text{ finite}\}$.
\end{define}

\begin{lemma}\rm\label{lem-bdprocessisrun}
  Let $N$ be a net and $\PP$ a process thereof.

  Then $\BDify{\PP}$ is a BD-run.
\end{lemma}
\begin{proof}
  See \cite[Lemma 1]{glabbeek11ipl}.
\qed
\end{proof}

\noindent
This immediately yields an injective function from BD-processes to BD-runs, since
by \refdf{swapeq}, $\PP \swapeq^\infty \QQ$ iff $\BDify{\PP} = \BDify{\QQ}$.
For countable nets (i.e.\ nets with countably many places and
transitions), this function is even a bijection.

\begin{lemma}\rm\label{lem-bdprocessexists}
  Let \begin{TR}$N=(S,T,F,M_0)$\end{TR}\begin{CONCUR}$N$\end{CONCUR}
  be a countable net and $\R$ a BD-run of $N$.

Then $\R$ is countable and
there exists a process $\PP$ of $N$ such that $\R = \BDify{\PP}$.
\end{lemma}
\begin{TR}
\begin{proof}
  Up to isomorphism there is only one GR-process $((\SS, \TT, \FF,
  \MM_0), \pi)$ of $N$ with $|\TT| = 0$.  Furthermore, as $N$ is
  countable, up to isomorphism there are only countably many with
  $|\TT| = 1$, countably many with $|\TT| = 2$ and so on. Given
  that isomorphic GR-processes are swapping equivalent, there are only
  countably many BD-processes with any given finite number of
  transitions. Hence we can enumerate all finite BD-processes of
  $N$. As $\R$ contains only finite BD-processes of $N$, it, too, must
  be countable.

  We construct a sequence of processes
  $\PP_i = ((\SS_i, \TT_i, \FF_i, {\MM_0}), \pi_i)$.
  We start with $\PP_0 = ((\SS_0, \varnothing, \varnothing, \MM_0), \pi_0)$
  where $\SS_0:=\{(s,i)\in S\times\bbbn \mid i<M_0(s)\}$, $\MM_0(s,i)=1$ and $\pi_0(s,i)=s$.
  As $\BD{\PP_0}$ is a prefix of every finite BD-process of $N$, $\BD{\PP_0} \in \R$.

  We can enumerate the elements of $\R$ as $\QQ_1, \QQ_2, \ldots$.
  Now given a process $\PP_i$ with $\BD{\PP_i} \in \R$, consider $\BD{\QQ_i} \in \R$.
  As $\R$ is directed, there exists a $\BD{\PP'} \in \R$ with
  $\BD{\PP_i} \leq \BD{\PP'} \wedge \BD{\QQ_i} \leq \BD{\PP'}$ which is to say
  there exists some $\PP_{i+1}$ with $\PP_i \leq \PP_{i+1} \swapeq^* \PP'$,
  and some $\QQ'$ with $\QQ_i \leq \QQ' \swapeq^* \PP'$.
  We have $\BD{\PP_{i+1}} =\BD{\PP'} \in \R$.

  The limit $((\bigcup_{i=0}^\infty\SS_i, \bigcup_{i=0}^\infty\TT_i,
  \bigcup_{i=0}^\infty\FF_i, \MM_0), \bigcup_{i=0}^\infty\pi_i)$
  of the $\PP_i$ is the $\PP$ we had to find.
  We need to show that $\R = \BDify{\PP}$.

  Take any element $\BD{\QQ_i}$ of $\R$. Per construction,
  $\BD{\QQ_i} \leq \BD{P_{i+1}}$ and $P_{i+1} \leq P$, so
  $\BD{\QQ_i} \in \BDify{\PP_i}$.
  Hence $\R \subseteq \BDify{\PP}$.

  Now take any $\BD{\QQ} \in \BDify{\PP} = \mathord\downarrow\{\BD{\PP'} \mid
    \PP' \leq \PP, \PP' \text{ finite}\}$.
  Then there exist some finite $\QQ'$ such that
  ${\BD{\QQ} \leq \BD{\QQ'}} \wedge {\QQ' \leq \PP}$.
  The process $\QQ'$ has finitely many transitions.
  Hence there exists some $i$ such that all of these transitions occur in
  $\PP_i$ and as $\QQ' \leq \PP$ then also $\QQ' \leq \PP_i$.
  Since $\BD{\PP_i} \in \R$ and $\R$ is prefix closed, we have $\BD{\QQ'} \in \R$
  and $\BD{\QQ} \in \R$.
  \qed
\end{proof}
\end{TR}
\reflem{bdprocessexists} does not hold for uncountable nets,
as witnessed by the counterexample in \reffig{uncountable}.
This net $N$ has a transition $t$ for each real number $t\inp\bbbr$.
Each such transition has a private preplace $s_t$ with $M_0(s_t)=1$
and $F(s_t,t)=1$, which ensures that $t$ can fire only
once. Furthermore there is one shared place $s$ with $M_0(s)=2$ and a
loop $F(s,t)=F(t,s)=1$ for each transition $t$. There are no other
places, transitions or arcs besides the ones mentioned above.

Each GR-process of $N$, and hence also each BD-process, has only
countably many transitions. Yet, any two GR-processes firing the same
finite set of transitions of $N$ are swapping equivalent, and the set
of all finite BD-processes of $N$ constitutes a single BD-run
involving all transitions.

\begin{figure}
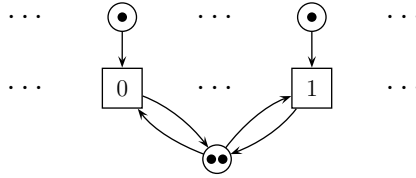

\vspace{-2ex}
  \begin{center}
    \psscalebox{0.9}{
    \begin{petrinet}(8,4)
      \P(2,3.5):p0;
      \P(6,3.5):p1;

      \p(4,0.5):pall;
      \pscircle*(3.88,0.5){0.1}
      \pscircle*(4.12,0.5){0.1}

      \t(2,2):t0:$0$;
      \t(6,2):t1:$1$;

      \a p0->t0;
      \a p1->t1;

      \A pall->t0; \A t0->pall;
      \A pall->t1; \A t1->pall;
      
      \rput(0,3.5){\psscalebox{2}{$\cdots$}}
      \rput(8,3.5){\psscalebox{2}{$\cdots$}}
      \rput(4,3.5){\psscalebox{2}{$\cdots$}}
      \rput(0,2.0){\psscalebox{2}{$\cdots$}}
      \rput(8,2.0){\psscalebox{2}{$\cdots$}}
      \rput(4,2.0){\psscalebox{2}{$\cdots$}}
    \end{petrinet}
    }
  \end{center}
  \vspace{-3ex}
  \caption{A net with no maximal GR-process, but with a maximal BD-run.}
  \label{fig-uncountable}
  \vspace{-2ex}
\end{figure}

\noindent
We now show that the mapping \textit{BD} respects the ordering of processes.

\begin{lemma}\rm\label{lem-procleqimplrunleq}
  Let $N$ be a net, and $\PP$ and $\PP'$ two GR-processes of $N$.

  If $\PP \leq \PP'$ then $\BDify{\PP} \subseteq \BDify{\PP'}$.
\end{lemma}
  \begin{TR}
$\begin{array}{@{}l@{\hspace{23pt}}l@{\hspace{23pt}}r@{}}
\textit{Proof.} &
    \PP \leq \PP'
    \implies \\&
    \{\QQ \mid \QQ \leq \PP,\, \QQ \text{ finite}\} \subseteq
      \{\QQ' \mid \QQ' \leq \PP',\, \QQ' \text{ finite}\}
    \implies \\&
    \{\BD{\QQ} \mid \QQ \leq \PP,\, \QQ \text{ finite}\} \subseteq
      \{\BD{\QQ'} \mid \QQ' \leq \PP',\, \QQ' \text{ finite}\}
    \implies \\&
    \mathord\downarrow\{\BD{\QQ} \mid \QQ \leq \PP,\, \QQ \text{ finite}\} \subseteq
      \mathord\downarrow\{\BD{\QQ'} \mid \QQ' \leq \PP',\, \QQ' \text{ finite}\}
    \implies \\&
    \BDify{\PP} \subseteq \BDify{\PP'}.
& \mbox{\qed}
\end{array}$
  \end{TR}

\section{Conflicts in place/transition systems}\label{sec-conflict}

We recall the canonical notion of conflict introduced in
\cite{goltz86howmany}.

\begin{define}{
  Let $N \mathbin= (S, T, F, M_0)$ be a net and $M \in
  \bbbn^S\!$.
}\label{df-semanticconflict}
\item
  A finite, non-empty multiset $G \in \bbbn^T$ is in
  \defitem{(semantic) conflict} in $M$ iff\\
  $(\forall t \in G. M\goesto[G \restrictedto \{t\}]) \wedge \neg M\goesto[G]$.
\item
  $N$ is \defitem{(semantic) conflict-free} iff
  no finite, non-empty multiset $G \in \bbbn^T$ is in semantic conflict in any
  $M$ with $M_0 \goesto[] M$.
\item
  $N$ is \defitem{binary-conflict-\!-free} iff
  no multiset $G \in \bbbn^T$ with $|G| = 2$ is in semantic conflict in any
  $M$ with $M_0 \goesto[] M$.
\end{define}

\begin{trivlist}
\item[\hspace{\labelsep}\bf Remark:]
In a net $(S,T,F,M_0)$ with $S=\{s\}$, $T=\{t,u\}$, $M_0(s)=1$ and
$F(s,t)=F(s,u)=1$, the multiset $\{t,t\}$ is not enabled in $M_0$.
For this reason the multiset $\{t,t,u\}$ does not count as being in
conflict in $M_0$, even though it is not enabled. However, its subset
$\{t,u\}$ is in conflict.
\end{trivlist}

\begin{TR}
\begin{lemma}\rm\label{lem-cfdiamond}
  Let $N\mathbin=(S,T,F,M_0)$ be a binary-conflict-\!-free net,
  $a,b\inp T$ with $a\mathbin{\neq} b$, and $P,P',Q$ be finite GR-processes of $N$.

  If $P\goesto[a]P'$ and $P\goesto[b]Q$
  then $\widehat P \goesto[\{a,b\}]$ and $\exists Q'. P'\goesto[b]Q' \wedge \BD{Q}\goesto[a]\BD{Q'}$.
\end{lemma}
\begin{proof}
  Suppose $P\goesto[a]P'$ and $P\goesto[b]Q$ with  $a\neq b$.
  Then $M_0\goesto[] \widehat P$ by \refobs{bisimulation}(d).
  Moreover, $\widehat P \goesto[a] \widehat P'$ and $\widehat P\goesto[b]\widehat Q$
  by \refobs{bisimulation}(b).
  Hence, as $N$ is binary-conflict-\!-free, $\widehat P \goesto[\{a,b\}]$.
  Therefore $\widehat P' \goesto[b]M$ for some $M$. Using \refobs{bisimulation}(c),
  there exists a GR-process $Q'$ with $P'\goesto[b]Q'$ and $\widehat Q' = M$.
  This $Q'$ can be chosen in such a way that the $b$-transition uses no tokens that
  are produced by the preceding $a$ transition.
  We then obtain that $P\goesto[b]Q''\goesto[a]Q'$ for some $Q''$.
  By \refobs{swaptrans}(a), $Q\swapeq^* Q''$, and hence $\BD{Q}\goesto[a]\BD{Q'}$.
\qed
\end{proof}
\end{TR}

\noindent
We proposed in \cite{glabbeek11ipl} a class of \mbox{P\hspace{-1pt}/T} systems
where the structural definition of conflict in terms of shared preplaces, as
often used in Petri net theory, matches the semantic definition of conflict as
given above. We called this class of nets \defitem{structural conflict
nets}. For a net to be a structural conflict net, we require that two
transitions sharing a preplace will never occur both in one step.

\begin{define}{
  Let $N \mathbin= (S, T, F, M_0)$ be a net.
}\label{df-structuralconflict}
\item[]
  $N$ is a \defitem{structural conflict net} iff
  $\forall t, u.
    (M_0 \goesto[]\;\goesto[\{t, u\}]) \implies
    \precond{t} \cap \precond{u} = \emptyset$.
\end{define}
Note that this excludes self-concurrency from the possible behaviours in a
structural conflict net: as in our setting every transition has at least one
preplace, $t = u$ implies $\precond{t} \cap \precond{u} \ne \emptyset$.  Also
note that in a structural conflict net a non-empty, finite multiset $G$ is in
conflict in a marking $M$ iff $G$ is a set and two distinct transitions in $G$
are in conflict in $M$. Hence a structural conflict net is conflict-free if
and only if it is binary-conflict-\!-free.  Moreover, two transitions enabled
in $M$ are in (semantic) conflict iff they share a preplace.

\section{A conflict-free net has exactly one maximal run}
\label{sec-finiteruns}

\noindent
In this section, we recapitulate results from \cite{best87both},
giving an alternative characterisation of runs of a net in terms of
firing sequences. We use an adapted notation and terminology and a
different treatment of infinite runs, as in \cite{glabbeek11ipl}.
As a main result of the present paper, we then prove in this
setting that a conflict-free net has exactly one maximal run. In the
following section, this result will be transferred to BD-processes.

The behaviour of a net can be described not only by its processes, but also by
its firing sequences. The imposed total order on transition firings abstracts
from information on causal dependence, or concurrency, between transition
firings.  To retrieve this information we introduce an \defitem{adjacency}
relation on firing sequences, recording which interchanges of transition
occurrences are due to semantic independence of transitions. Hence adjacent
firing sequences represent the same run of the net. We then define
\defitem{FS-runs} in terms of the resulting equivalence classes of firing
sequences.

\begin{define}{
  Let $N = (S, T, F, M_0)$ be a net, and $\sigma, \rho \in \FS(N)$.
}\label{df-connectedto}
\item
  $\sigma$ and $\rho$ are \defitem{adjacent}, $\sigma \leftrightarrow \rho$,
  iff $\sigma = \sigma_1 t u \sigma_2$, $\rho = \sigma_1 u t \sigma_2$ and
  $M_0 \goesto[\sigma_1]\goesto[\{t, u\}]$.
\item
  We write $\connectedto$ for the reflexive and transitive closure of $\adjacent$,
  and $[\sigma]$ for the $\connectedto$-equivalence class of a firing sequence $\sigma$.
\end{define}

\noindent
Note that $\connectedto$-related firing sequences contain the same
(finite) multiset of transition occurrences.
When writing $\sigma \connectedto \rho$ we implicitly claim that $\sigma, \rho \in \FS(N)$.
Furthermore $\sigma \connectedto \rho \wedge \sigma\mu \in \FS(N)$
implies $\sigma \mu \connectedto \rho \mu$ for all $\mu \in T^*$.

The following definition introduces the notion of \defitem{partial}
FS-run which is a formalisation of the intuitive concept of a finite,
partial run of a net.

\begin{define}{
  Let $N$ be a net and $\sigma, \rho \in \FS(N)$.
}\label{df-partialrun}
\item
  A \defitem{partial FS-run} of $N$ is an $\connectedto$-equivalence class
  of firing sequences.
\item
  A partial FS-run $[\sigma]$ is a \defitem{prefix} of another partial FS-run
  $[\rho]$, notation \mbox{$[\sigma]\leq[\rho]$}, iff
  $\exists \mu. \sigma \leq \mu \connectedto\! \rho$.
\end{define}
Note that $\sigma' \connectedto\! \sigma \leq \mu$ implies
$\exists \mu'. \sigma' \leq \mu' \connectedto \mu$;
thus the notion of prefix is well-defined, and a partial order.

Similar to the construction of BD-runs out of finite BD-processes,
the following concept of an FS-run extends the notion of a partial
FS-run to possibly infinite runs, in such a way that an FS-run is
completely determined by its finite approximations.

\begin{define}{
  Let $N$ be a net.
}\label{df-run}
\item[]
  An \defitem{FS-run} of $N$ is a non-empty, prefix-closed and directed set of
  partial FS-runs of $N$.
\end{define}
There is a bijective correspondence between partial FS-runs and the
finite elements in the class of FS-runs, just as in the case of
BD-runs in Section~\ref{sec-semantics}.
Much more interesting however is the following bijective
correspondence between BD-runs and FS-runs.

\begin{theorem}\rm\label{thm-fsbd}
  There exists a bijective function $\Pi$ from FS-runs to BD-runs
  such that $\Pi(\R) \subseteq \Pi(\R')$ iff $\R \subseteq \R'$.
\end{theorem}
\begin{proof}
  See \cite{glabbeek11ipl}, in particular the remarks at the end of Section 5.
\qed
\end{proof}

\begin{TR}
\noindent
We use the relations between firing sequences up to $\connectedto$ and
finite GR-processes up to $\swapeq^*$, as examined in
\cite{best87both}, to establish the following variant of
\reflem{cfdiamond}, which we will need in the next section.

\begin{lemma}\rm\label{lem-scdiamond}
  Let $N\mathbin=(S,T,F,M_0)$ be a structural conflict net,
  $a,b\inp T$ with $a\mathbin{\neq} b$,
  $\R$ be a BD-run of $N$, and $P,P',Q\in^2\R$.
  (Here $X\mathbin{\in^2} Z$ stands for $\exists Y. X\inp Y\inp Z$.)

  If $P\!\goesto[a]\!P'$ and $P\!\goesto[b]\!Q$ then $\widehat P \!\goesto[\{a,b\}]$
  and $\exists Q'\!\in^2\!\R. P'\!\goesto[b]\!Q' \wedge \BD{Q}\!\goesto[a]\!\BD{Q'}$.
\end{lemma}

\begin{proof}
  Suppose $P\goesto[a]P'$ and $P\goesto[b]Q$ with  $a\neq b$.
  Then $M_0\goesto[] \widehat P$ by \refobs{bisimulation}(d).
  Let $\sigma\in T^*$ be such that $M_0\goesto[\sigma] \widehat P$.
  As runs are directed, there is a finite BD-process $\BD{U}\in\R$ with
  $\BD{P'}\leq \BD{U}$ and $\BD{Q}\leq \BD{U'}$.
  Hence there must be a sequence 
  $P'=P_1\goesto[a_1] P_2 \goesto[a_2] \ldots \goesto[a_{k}] P_{k+1}$
  with $P_{k+1}\in\BD{U}$, and, similarly, a sequence
  $Q=Q_1\goesto[b_1] Q_2 \goesto[b_2] \ldots \goesto[b_\ell] Q_{\ell+1}$
  with $Q_{\ell+1}\in\BD{U}$. Let $a_0:= a$ and $\rho:= a_0 a_1 \cdots
  a_k$; likewise let $b_0:=b$ and $\mu:= b_0 b_1 \cdots b_\ell$. By Theorem 3 in
  \cite{glabbeek11ipl} it follows that $\sigma \rho, \sigma\mu \in
  \FS(N)$ and $\sigma\rho \connectedto \sigma\mu$.
  By \refdf{connectedto}, $\sigma \rho$ and $\sigma \mu$ must
  contain the same multiset of  transitions. So $b=a_h$ for some
  $1\leq h\leq k$; we take $h$ minimal, so that $b\neq a_j$ for $0\leq j < h$.

  Let $Q'_h:=P_{h+1}$. Since $P_{h+1} \leq P_{k+1}\in^2\R$ and $\R$ is
  prefix-closed, we have $Q'_h\in^2\R$.
  Working our way down from $h\mathord-1$ to $0$, we construct
  for any $j$ with $0\leq j < h$ 
  a $Q'_j\in^2 \R$ with $P_j \goesto[b] Q'_j \goesto[a_j] Q'_{j+1}$.
  Suppose we already have $Q'_{j+1}$. Then $P_j\goesto[a_j]P_{j+1}\goesto[b]Q'_{j+1}$.

  Somewhere in the sequence $\sigma \rho = \nu_1
  \leftrightarrow \nu_2 \leftrightarrow \cdots \leftrightarrow \nu_n =
  \sigma \mu$ the transitions $a_j$ and $b$ must be exchanged,
  i.e.\ $\nu_i=\nu'a_jb\nu''\leftrightarrow\nu'ba_j\nu''=\nu_{i+1}$.\vspace{1pt}
  Thus there is a marking $M$ with $M_0\goesto[\nu']M\goesto[\{a_j,b\}]$.
  Since $N$ is a structural conflict net, $\precond{a_j}\cap\precond{b}=\emptyset$.
  This immediately yields a $Q'_j$ with $P_j \goesto[b] Q'_j \goesto[a_j] Q'_{j+1}$.
  Since $Q'_{j+1}\in^2\R$ and $\R$ is prefix-closed, we have $Q'_j\in^2\R$.

  Finally, let $Q':= Q'_1$. Then $P'\mathop=P_1\goesto[b]Q'_1\mathop=Q'$ and
  $P \mathop= P_0\goesto[b]Q'_0\goesto[a]Q'_1$.
  \refobs{swaptrans}(a) yields $Q\swapeq^* Q'_0$. Hence $\BD{Q}\goesto[a]\BD{Q'}$.
\qed
\end{proof}
\end{TR}

\noindent
We now show that a conflict-free net has exactly one maximal run. As we have a
bijective correspondence, it does not matter which notion of run we use here
(FS-run or BD-run).  We prove an even stronger result, using
binary-conflict-\!-free instead of conflict-free.  In preparation we need the
following lemma\begin{TR}s\end{TR}.

\begin{TR}
\begin{lemma}\rm\label{lem-conflictfreeswap}
  Let $N = (S, T, F, M_0)$ be a binary-conflict-\!-free net,
  $\sigma t, \sigma u \inp \FS(N)$ with $t, u \inp T$, and $t \mathbin{\ne} u$.

  Then $\sigma tu, \sigma ut \in \FS(N)$ and
  $\sigma tu \connectedto \sigma ut$.
\end{lemma}
\begin{proof}
  As we have unlabelled transitions, $\sigma$ leads to a unique marking.
  From $M_0 \goesto[\sigma t]{} \wedge M_0 \goesto[\sigma u]$ we thus
  have that an $M_1$ exists with $M_0 \goesto[\sigma]{} M_1 \wedge
  M_1 {\goesto[t]} \wedge M_1 \goesto[u]$. Due to binary-conflict-\!-freeness
  then also $M_1 \goesto[\{t, u\}]$.
  Hence $M_0 \goesto[\sigma]\goesto[\{t, u\}]$, so
  $\sigma tu, \sigma ut \in \FS(N)$ and
  $\sigma tu \connectedto \sigma ut$.
  \qed
\end{proof}

\begin{lemma}\rm\label{lem-closingdiamondwithout}
  Let $N = (S, T, F, M_0)$ be a binary-conflict-\!-free net,
  $\sigma t, \sigma \rho \inp \FS(N)$ with $t \inp T$, $\sigma,\rho\in T^*$,
  and $t \mathbin{\notin} \rho$.

  Then $\sigma t\rho, \sigma\rho t \in \FS(N)$ and
  $\sigma t\rho \connectedto \sigma \rho t$.
\end{lemma}
\begin{proof}
  Via induction on the length of $\rho$.

  If $\rho \mathbin= \epsilon$, $\sigma t\in\FS(N)$ trivially implies
  $\sigma \epsilon t, \sigma t \epsilon\in\FS(N)$ and
  $\sigma \epsilon t \connectedto \sigma t \epsilon$.

  For the induction step take $\rho \mathbin{:=} u\mu$ (with $u \ne t$).
  With $\sigma t, \sigma u\in\FS(N)$ and \reflem{conflictfreeswap} also
  $\sigma u t\in\FS(N)$ and $\sigma t u \connectedto \sigma u t$.
  Together with $\sigma u \mu\in\FS(N)$, the induction assumption then gives us
  $\sigma u t \mu \in \FS(N)$ and
  $\sigma u t \mu \connectedto \sigma u \mu t =
  \sigma \rho t$. With $\sigma u t \connectedto \sigma t u$ also
  $\sigma u t \mu \connectedto \sigma t u \mu =
  \sigma t \rho$ and $\sigma\rho t,\,\sigma t \rho\in \FS(N)$.
  \qed
\end{proof}

\begin{lemma}\rm\label{lem-closingdiamondwith}
  Let $N = (S, T, F, M_0)$ be a binary-conflict-\!-free net,
  $\sigma, \rho_1, \rho_2 \in T^*$, $t \in T$, $t \notin \rho_1$.

  If $\sigma t \in\FS(N) \wedge \sigma \rho_1 t \rho_2 \in\FS(N)$
  then $\sigma t \rho_1\rho_2\in\FS(N) \wedge
  \sigma t \rho_1 \rho_2 \connectedto \sigma \rho_1 t \rho_2$.
\end{lemma}
\begin{proof}
  Applying \reflem{closingdiamondwithout} with
  $\sigma t\inp\FS(N) \wedge \sigma \rho_1\inp\FS(N)$ 
  we get $\sigma t \rho_1\inp\FS(N) \wedge
  \sigma t \rho_1 \connectedto \sigma \rho_1 t$.
  Since $\sigma \rho_1 t \rho_2\in\FS(N)$ the latter yields
  $\sigma t \rho_1 \rho_2 \connectedto \sigma \rho_1 t \rho_2$
  and thus $\sigma t \rho_1 \rho_2\in\FS(N)$.
  \qed
\end{proof}
\end{TR}

\begin{lemma}\rm\label{lem-closingdiamond}
  Let $N$ be a binary-conflict-\!-free net.
  \vspace{2pt}

  If $\sigma,\sigma'\inp\FS(N)$ then
  $\exists \mu, \mu'. \sigma \mu\inp\FS(N) \wedge
  \sigma' \mu'\in\FS(N) \wedge \sigma\mu \connectedto \sigma'\mu'$.
\end{lemma}
  \begin{TR}
\begin{proof}
  Via induction on the length of $\sigma$.

  If $\sigma = \epsilon$ we take $\mu = \sigma'$ and $\mu' = \epsilon$.

  For the induction step we start with
  $$\sigma, \sigma'\in\FS(N) \implies
  \exists \mu, \mu'. \sigma\mu\in\FS(N) \wedge \sigma'\mu'\in\FS(N)
  \wedge \sigma\mu \connectedto \sigma'\mu'$$ and need to show that
  $$\sigma t, \sigma'\in\FS(N) \implies
  \exists \bar{\mu}, \bar{\mu}'.
  \sigma t\bar{\mu}\in\FS(N) \wedge \sigma'\bar{\mu}'\in\FS(N)
  \wedge \sigma t\bar{\mu} \connectedto \sigma'\bar{\mu}'\trail{.}$$

  If $t \inp \mu$, $\mu$ must be of the form $\mu_1 t \mu_2$ with $t \notin \mu_1$.
  We then take $\bar{\mu} := \mu_1 \mu_2$ and $\bar{\mu}' := \mu'$.
  By \reflem{closingdiamondwith} we find
  $\sigma t \mu_1 \mu_2\in\FS(N)$, i.e.\
  $\sigma t \bar{\mu}\in\FS(N)$.
  $\sigma' \bar{\mu}'\in\FS(N)$ is already contained in the induction
  assumption.  Per \reflem{closingdiamondwith} $\sigma t \bar{\mu} =
  \sigma t \mu_1\mu_2 \connectedto \sigma \mu_1 t \mu_2 =
  \sigma \mu$. From the induction assumption we obtain
  $\sigma \mu \connectedto \sigma' \mu' = \sigma' \bar{\mu}'$.

  If $t \mathbin{\notin} \mu$, we take $\bar{\mu} := \mu$ and $\bar{\mu}' := \mu't$.
  By \reflem{closingdiamondwithout} we find that
  $\sigma t \mu, \sigma \mu t\in\FS(N)$, i.e.\ also
  $\sigma t \bar{\mu}\in\FS(N)$. From $\sigma \mu t\in\FS(N)$
  and $\sigma \mu \connectedto \sigma' \mu'$ follows that
  $\sigma' \mu' t\in\FS(N)$, i.e.\ $\sigma' \bar{\mu}'\in\FS(N)$.
  Also by \reflem{closingdiamondwithout} we find
  $\sigma t \bar{\mu} = \sigma t \mu \connectedto \sigma \mu t$.
  From the induction assumption  we obtain
  $\sigma \mu t \connectedto \sigma' \mu' t= \sigma' \bar{\mu}'$.
  \qed
\end{proof}
  \end{TR}

\begin{theorem}\rm\label{thm-onemaximalrun}
  Let $N$ be a binary-conflict-\!-free net.

  There is exactly one maximal FS-run of $N$.
\end{theorem}
\begin{proof}
  Let $\R = \{[\sigma] \mid \sigma \text{ is a finite firing sequence of } N\}$.
  We claim that $\R$ is said maximal FS-run of $N$.

  First we show that $\R$ is prefix closed and directed, and thus indeed an FS-run.

  Take any $[\rho] \leq [\sigma] \in \R$. Then by definition of $\leq$,
  $\exists \nu. \rho \leq \nu \wedge \nu \connectedto \sigma$. We need to show
  that $[\rho] \in \R$, i.e.\ that $\rho$ is a firing sequence of $N$.
  Since $\sigma$ is a firing sequence of $N$ and $\nu \connectedto \sigma$,
  $\nu$ is also a firing sequence of $N$. Together with $\rho \leq \nu$
  follows that $\rho$, too, is a firing sequence of $N$.
  Thus $\R$ is prefix closed.

  To show directedness, let $[\sigma], [\rho] \in \R$.
  We need to show that $\exists [\nu] \inp \R.
  [\sigma] \leq [\nu]\linebreak[3] \wedge [\rho] \leq [\nu]$,
  or with the definitions of $\leq$ and $[~]$ expanded,
  $\exists \nu.
    (\exists \alpha. \sigma \leq \alpha \connectedto \nu\linebreak[3] \wedge
     \exists \beta. \rho \leq \beta \connectedto \nu)$.
  We now apply \reflem{closingdiamond} to $\sigma, \rho\in\FS(N)$,
  obtaining $\mu$ and $\mu'$ as mentioned in that
  lemma, and take $\alpha = \sigma\mu$ and $\beta = \rho\mu'$.
  Then \reflem{closingdiamond} gives us $\alpha \connectedto \beta$ and
  we take $\nu = \alpha$.
  Thus $\R$ is directed.

  Finally we show that $\R$ is maximal. Take any run $\R'$ of $N$.
  Then $\R' \subseteq \R$ by definition of $\R$, hence $\R$ is maximal.
  \qed
\end{proof}

\section{BD-processes fit structural conflict nets}
\label{sec-results}

In this section we show that BD-processes are adequate as abstract
processes for the subclass of structural conflict nets.

In \cite{glabbeek11ipl} we have shown that a semantic conflict in a structural conflict net
always gives rise to multiple maximal GR-processes even up to swapping equivalence.

\begin{theorem}\rm\label{thm-structuralconflictbdprocsoneway}
  Let $N$ be a structural conflict net.

  If $N$ has only one maximal GR-process up to $\swapeq^\infty$ then $N$ is conflict-free.
\end{theorem}
\begin{proof}
  Corollary~1 from \cite{glabbeek11ipl}.
\qed
\end{proof}

\noindent
We conjectured in \cite{glabbeek11ipl} that, for countable nets,
also the reverse direction holds, namely that a countable
conflict-free structural conflict net has exactly one maximal
GR-process up to $\swapeq^\infty$.

In \refsec{finiteruns} we have already shown that a corresponding result holds for runs
instead of processes. We will now transfer this result to BD-processes, and
hence prove the conjecture.

We proceed by investigating three notions of maximality for BD-processes;
they will turn out to coincide for structural conflict nets.

\begin{definition}\rm\label{df-bdprocmax}~\vspace{-1.5ex}
  \begin{itemize}
    \item
      A BD-process $\BDinf{P}$ is \defitem{weakly maximal} (or a maximal
      GR-process up to $\swapeq^\infty$), iff some $P' \mathbin\in \BDinf{P}$
      is maximal (in the GR-process sense).
    \item
      A BD-process $\BDinf{P}$ is \defitem{maximal} iff
      $\forall P' \in \BDinf{P}\,\forall Q. (P' \leq Q \Rightarrow P' \swapeq^\infty Q)$.
    \item
      A BD-process $\BDinf{P}$ is \defitem{run-maximal} iff
      the BD-run $\BDify{P}$ is maximal. % (in the set theoretic sense).
  \end{itemize}
\end{definition}
The first notion is the simplest way of inheriting the notion of
maximality of GR-process by BD-processes, whereas the last one
inherits the notion of maximality from BD-runs.
The middle notion is the canonical notion of maximality with respect
to a natural order on BD-process, defined below.

\begin{definition}\rm\label{df-bdprocorder}
  Let $N$ be a net.

  We define a relation $\preceq$ between BD-processes, via
  $$
  \BDinf{P} \preceq \BDinf{Q} :\Leftrightarrow
  \exists P' \swapeq^\infty P\;
  \exists Q' \swapeq^\infty Q. P' \leq Q'~,
  $$

  and construct an order between BD-processes via
  $$
  \BDinf{P} \leq \BDinf{Q} :\Leftrightarrow
  \BDinf{P} \preceq^+ \BDinf{Q}~.
  $$
\end{definition}
By construction, the relation $\leq$ is reflexive and transitive
(even though $\preceq$ in general is not transitive).
\reflem{procleqimplrunleq} yields that it also is antisymmetric, and
hence a partial order. Namely, if $\BDinf{P} \leq \BDinf{Q}$ and
$\BDinf{Q} \leq \BDinf{P}$, then $\BDify{P}=\BDify{Q}$, so
$P \swapeq^\infty Q$, implying $\BDinf{P} = \BDinf{Q}$.

Now maximality according to \refdf{bdprocmax} is simply maximality
w.r.t.\ $\leq$:
\begin{center}
$\BDinf{P}$ is maximal iff $\nexists \BDinf{P'}.
 \BDinf{P} \leq \BDinf{P'} \wedge \BDinf{P} \ne \BDinf{P'}$.
\end{center}
The following lemma tells how the above notions of maximality form a hierarchy.

\begin{lemma}\rm\label{lem-maxorder}
  Let $N$ be a net and $P$ a process thereof.
  \vspace{-1ex}
  \begin{enumerate}
    \item If $\BDinf{P}$ is run-maximal, it is maximal.
    \item If $\BDinf{P}$ is maximal, it is weakly maximal.
  \end{enumerate}
\end{lemma}
\begin{proof}
  ``1'':
  This follows since $\BDinf{P}\leq\BDinf{Q}
  \Rightarrow \BDify{P}\subseteq\BDify{Q}$ by \reflem{procleqimplrunleq}.

\begin{TR}
  Alternatively, assume $\BDify{P}$ is maximal.
  Take $P'\inp\BDinf{P}$ and $Q$ such that $P'\leq Q$.
  Then $\BDify{P'} \subseteq \BDify{Q}$ by
  \reflem{procleqimplrunleq}, but $\BDify{P'} = \BDify{P}$ which is maximal.
  Hence $\BDify{Q} = \BDify{P}$ and $P \swapeq^\infty Q$.
\end{TR}

  ``2'':
  Assume $\BDinf{P}$ is maximal.
  By Lemma 2 in \cite{glabbeek11ipl}, which follows via Zorn's
  Lemma, there exists some maximal $Q$ with $P \leq Q$.
  Since $\BDinf{P}$ is maximal we have $Q \swapeq^\infty P$ and $Q$ is a
  maximal process within $\BDinf{P}$.
  \qed
\end{proof}

\begin{figure}
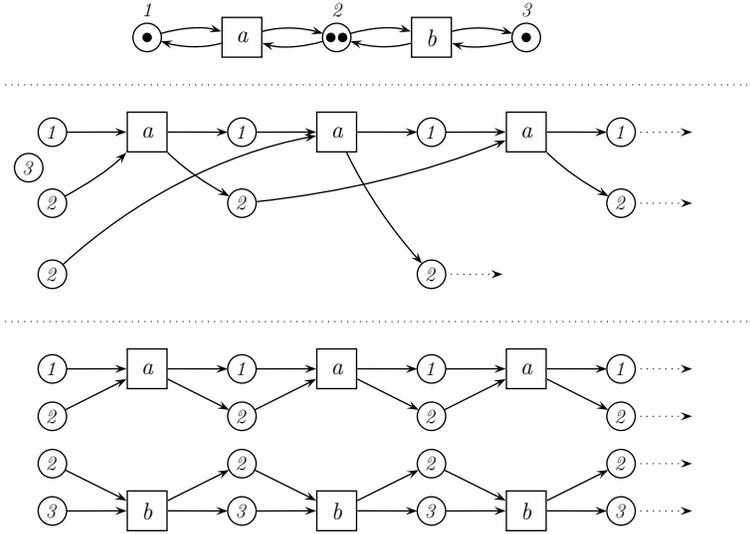

\vspace{-6ex}
  \begin{center}
    \psscalebox{0.9}{
    \begin{petrinet}(16,11)
      \Qt(3,10):p:1;
      \t(5,10):a:a;
      \qt(7,10):q:2;
      \t(9,10):b:b;
      \Qt(11,10):r:3;
      \A p->a; \A a->p;
      \A q->a; \A a->q;
      \A q->b; \A b->q;
      \A r->b; \A b->r;
      \pscircle*(6.875,10){0.1}
      \pscircle*(7.125,10){0.1}

      \psline[linestyle=dotted](0,9)(16,9)

      \s(0.5,7.25):r1:3;
      \s(1,8):p1:1;
      \s(1,6.5):q1:2;
      \s(1,5):q2:2;

      \t(3,8):a1:a;
      \t(7,8):a2:a;
      \t(11,8):a3:a;

      \s(5,8):p2:1;
      \s(9,8):p3:1;
      \s(13,8):p4:1;

      \s(5,6.5):q3:2;
      \s(9,5):q4:2;
      \s(13,6.5):q5:2;

      \a p1->a1; \a a1->p2; \a p2->a2; \a a2->p3; \a p3->a3; \a a3->p4;
      \B q1->a1; \B a1->q3; \B q3->a3; \B a3->q5;
      \A q2->a2; \B a2->q4;

      \psline[linestyle=dotted]{->}(13.3,8)(14.5,8)
      \psline[linestyle=dotted]{->}(13.3,6.5)(14.5,6.5)
      \psline[linestyle=dotted]{->}(9.3,5)(10.5,5)

      \psline[linestyle=dotted](0,4)(16,4)

      \s(1,3):pb1:1;
      \s(1,2):qb1:2;
      \s(1,1):qb2:2;
      \s(1,0):rb1:3;

      \t(3,3):ab1:a;
      \t(7,3):ab2:a;
      \t(11,3):ab3:a;

      \t(3,0):bb1:b;
      \t(7,0):bb2:b;
      \t(11,0):bb3:b;

      \s(5,3):pb2:1;
      \s(9,3):pb3:1;
      \s(13,3):pb4:1;

      \s(5,2):qb3:2;
      \s(9,2):qb4:2;
      \s(13,2):qb5:2;

      \s(5,1):qb6:2;
      \s(9,1):qb7:2;
      \s(13,1):qb8:2;

      \s(5,0):rb2:3;
      \s(9,0):rb3:3;
      \s(13,0):rb4:3;

      \a pb1->ab1; \a ab1->pb2; \a pb2->ab2; \a ab2->pb3; \a pb3->ab3; \a ab3->pb4;
      \a qb1->ab1; \a ab1->qb3; \a qb3->ab2; \a ab2->qb4; \a qb4->ab3; \a ab3->qb5;
      \a qb2->bb1; \a bb1->qb6; \a qb6->bb2; \a bb2->qb7; \a qb7->bb3; \a bb3->qb8;
      \a rb1->bb1; \a bb1->rb2; \a rb2->bb2; \a bb2->rb3; \a rb3->bb3; \a bb3->rb4;

      \psline[linestyle=dotted]{->}(13.3,3)(14.5,3)
      \psline[linestyle=dotted]{->}(13.3,2)(14.5,2)
      \psline[linestyle=dotted]{->}(13.3,1)(14.5,1)
      \psline[linestyle=dotted]{->}(13.3,0)(14.5,0)

    \end{petrinet}
    }
  \end{center}
  \caption{A net and two weakly maximal processes thereof.}
  \label{fig-weakmaximal}
\end{figure}

\noindent
The three notions of maximality are all distinct.
The first process depicted in \reffig{weakmaximal} is an example of a weakly
maximal BD-process that is not maximal.
Namely, the process itself cannot be extended (for none of the tokens
in place 2 will in the end come to rest), but the process is swapping
equivalent with the top half of the second process (using only one
of the tokens in place 2), which can be extended with the bottom half.

\begin{figure}
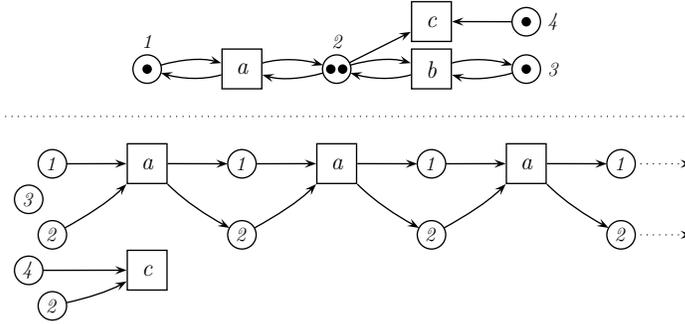

  \vspace{2ex}
  \begin{center}
    \psscalebox{0.9}{
    \begin{petrinet}(16,6)
      \Qt(3,5):p:1;
      \t(5,5):a:a;
      \qt(7,5):q:2;
      \t(9,5):b:b;
      \Q(11,5):r:3;
      \t(9,6):c:c;
      \Q(11,6):s:4;
      \A p->a; \A a->p;
      \A q->a; \A a->q;
      \A q->b; \A b->q;
      \A r->b; \A b->r;
      \a q->c; \a s->c;
      \pscircle*(6.875,5){0.1}
      \pscircle*(7.125,5){0.1}

      \psline[linestyle=dotted](0,4)(16,4)

      \s(0.5,2.25):r1:3;
      \s(1,3):p1:1;
      \s(1,1.5):q1:2;
      \s(1,0):q2:2;
      \s(0.5,0.75):s1:4;

      \t(3,3):a1:a;
      \t(7,3):a2:a;
      \t(11,3):a3:a;

      \t(3,0.75):c1:c;

      \s(5,3):p2:1;
      \s(9,3):p3:1;
      \s(13,3):p4:1;

      \s(5,1.5):q3:2;
      \s(9,1.5):q4:2;
      \s(13,1.5):q5:2;

      \a p1->a1; \a a1->p2; \a p2->a2; \a a2->p3; \a p3->a3; \a a3->p4;
      \B q1->a1; \B a1->q3; \B q3->a2; \B a3->q5;
      \B a2->q4; \B q4->a3;
      \a s1->c1; \B q2->c1;

      \psline[linestyle=dotted]{->}(13.3,3)(14.5,3)
      \psline[linestyle=dotted]{->}(13.3,1.5)(14.5,1.5)
    \end{petrinet}
    }
  \end{center}
  \caption{A net and a maximal process thereof.}
  \label{fig-onlymaximal}
\end{figure}

The process depicted in \reffig{onlymaximal} is an example of a BD-process
$\BDinf{P}$ which is maximal, but not run-maximal.
It is maximal, because no matter how it is swapped, at some point the
$c$-transition will fire, and after that the only token left in place 2
will be in use forever, making it impossible to extend the process
with any ($b$-)transition.
It is not run-maximal, as the set of all finite processes of $N$
constitutes a larger run.
Note that every two finite processes of $N$ mapping to the same multiset
of transitions are swapping equivalent.

The following lemmas show that for countable conflict-free nets
maximality and run-maximality coincide, and that for structural
conflict nets all three notions of maximality coincide.

\begin{lemma}\rm\label{lem-cfmaximals}
  Let $N$ be a countable binary-conflict-\!-free net, and $P$ be a GR-process of $N$.

\begin{TR}
  \vspace{-1.5ex}
  \begin{enumerate}
    \item[(1)] If $\BDify{P}$ is non-maximal, then
        $\exists P' \mathbin{\swapeq^\infty} P\, \exists Q.
        P' \mathbin\leq Q \wedge P' \mathbin{\not\swapeq^\infty} Q$.
    \item[(2)] If $\BDinf{P}$ is maximal, then $\BDinf{P}$ is run-maximal.
  \end{enumerate}
\end{TR}
\begin{CONCUR}
  If $\BDinf{P}$ is maximal, then $\BDinf{P}$ is run-maximal.
\end{CONCUR}
\end{lemma}
  \begin{TR}
\begin{proof}
  ``(1)'':
  Take $\R \supsetneq \BDify{P}$. Take a minimal
  $Q_0 \in^2 \R \mathop{\setminus} \BDify{P}$.
  $Q_0$ can be written as $((\SS'', \TT\,'', \FF\,'', \MM_0''), \pi'')$.
  Let $t$ be a maximal element in $\TT\,''$ with respect to $\FF\,''^+$.
  Then $Q_0 \upharpoonright (\TT\,'' \setminus \{t\}) =: Q'_0$ is a process
  and $Q'_0 \mathbin{\in^2} \BDify{P}$ (as otherwise $Q_0$ would not
  have been minimal).
  Hence there exists finite $P'_0, Q'$ such that $Q'_0 \leq Q' \swapeq P'_0 \leq P$.
  Moreover, there are $Q'_1,\ldots,Q'_n\in^2 \BDify{P}$ with $Q_n=Q'$
  and $Q'_{i-1}\goesto[a_i]Q'_i$ for $i=1,\ldots,n$.

  $\pi_Q(t)$ is some transition $b$ of $N$, so $Q'_0 \goesto[b] Q_0$.
  We now show by induction on $i\in\{1,\ldots,n\}$ that there are
  $Q_1,\ldots,Q_n\not\in^2 \BDify{P}$
  with $Q'_i\goesto[b]Q_i$ and \mbox{$\BD{Q_{i-1}}\goesto[a_i]\BD{Q_i}$} for $i\mathbin=1,\ldots,n$.
  Namely, given $Q_{i-1}$, as $Q_{i-1}\mathbin{\not\in^2}\BDify{P}$ we have
  $Q_{i-1}\not\swapeq^* Q'_i\in^2\BDify{P}$.
  Using that $Q'_{i-1} \goesto[a_i] Q''_i\swapeq^* Q'_i$ and $Q'_{i-1} \goesto[b] Q_{i-1}$,
  this implies $a_i\mathbin{\neq} b$ by \refobs{swaptrans}(a).
  Now \reflem{cfdiamond} yields a $Q_i$ such that 
  $Q'_{i} \goesto[b] Q_{i}$ and $\BD{Q_{i-1}} \goesto[a_i] \BD{Q_i}$.
  As $\BDify{P}$ is prefix closed, we have $Q_i\mathbin{\not\in^2}\BDify{P}$.

  Since $Q'_n\swapeq^* P'_0$ and $Q'_n\goesto[b] Q_{n}$,
  there is a $P_0$ with $P'_0\goesto[b] P_0$ and $P_0\swapeq^*Q_{n}$, using \refobs{swaptrans}(b).
  Hence $P_0\not\in^2\BDify{P}$.

  Let $P = ((\SS, \TT, \FF, {\MM_0}, \pi)$,
  $P'_0 = ((\SS', \TT\,', \FF\,', {\MM_0}, \pi')$
  and $N=(S,T,F,M_0)$.
  Enumerate the transitions in $\TT\setminus\TT\,'$ as $\{t_i \mid i\inp\bbbn\}$,
  such that if $t_i \FF^+ t_j$ then $i \mathbin< j$.
  This is always possible, since $N$ is countable 
  and $\{t \mid (t,u)\inp \FF^+\}$ is finite for all $u\inp \TT$.
  So there are $P'_i \leq P$ for $i\mathbin>0$ such that
  \mbox{$P'_0 \goesto[\pi(t_0)] P'_1 \goesto[\pi(t_1)] P'_2 \goesto[\pi(t_2)] \cdots$}.
  Exactly as above, by induction on $i$, there must be $P_1,P_2,\ldots\not\in^2 \BDify{P}$
  with $P'_{i+1}\goesto[b]P_{i+1}$ and
  \mbox{$\BD{P_{i}}\goesto[\pi(t_i)]\BD{P_{i+1}}$} for $i=0,\ldots,m$.
  Moreover, \plat{$\widehat P'_i \goesto[\{\pi(t_i),b\}]$} by \reflem{cfdiamond}.

  Let $\SS\!_b:=P_0'^\circ\setminus P_0^\circ$. Then $\pi(\SS\!_b) = \precond{b}$.
  With induction on $j$, for each transition $t_j$ pick a set 
  ${}^\circ t_j \subseteq \SS$ with $\pi({}^\circ t_j) =
  \precond{\pi(t_j)} \;( = \pi(\precond{t_j}))$ such that
  $${}^\circ t_j \subseteq {P_0'^\circ} \cup (\bigcup_{i<j} \postcond{t_j}) \setminus 
          (\bigcup_{i<j} {}^\circ t_i) \setminus \SS\!_b {\;.}$$
  Such a set always exists, since
$$
    \pi\left({P_0'^\circ} \cup (\bigcup_{i<j} \postcond{t_j}) \setminus 
          (\bigcup_{i<j} {}^\circ t_i) \right) =
    \widehat P'_0 + \sum_{i<j} \postcond{\pi(t_i)} - \sum_{i<j}\precond{\pi(t_i)} =
    \widehat P'_j \supseteq \precond{\pi(t_j)}+\precond{b}.$$
Let $\bar{P} = ((\SS, \TT, \FF\!_{\bar{P}}, {\MM_0}), \pi)$ with
  \begin{equation*}
    \FF\!_{\bar{P}}(x, y) := 
    \begin{cases}
      \FF(x, y)&\text{ if } (x \in \TT \wedge y \in \SS) \vee y \in \TT\,'\\
      ({}^\circ y)(x)&\text{ otherwise (i.e. } y \in \TT\setminus \TT\,'
      \ \wedge x \in \SS\text{)}.\\
    \end{cases}
  \end{equation*}
$\bar{P}$ is a process via the construction.
Namely, for all $s\in \SS$, $\precond{s}$ in $\bar{P}$ is the same as
in $P$, and hence $|\precond{s}| \leq 1$. Likewise, $\MM_0$ is unchanged.
We have $|\postcond{s}| \leq 1$ by construction, in particular because
${}^\circ{t_i} \cap {}^\circ{t_j}=\emptyset$ for $i<j$.
If $(t_i,t_j)\in \FF\!_{\bar{P}}^{\,+}$ then $i<j$, from which it can be
inferred that $\FF\!_{\bar{P}}$ is acyclic and $\{t \mid (t,u)\in
\FF\!_{\bar{P}}^{\,+}\}$ is finite for all $u\in \TT$.
The conditions $\pi(\MM_0)=M_0$, $\pi(\precond{t})=\precond{\pi(t)}$
and $\pi(\postcond{t})=\postcond{\pi(t)}$ hold for $\bar{P}$ because
they do for $P$, and we have $\pi({}^\circ t_j) = \precond{\pi(t_j)}$.

By construction, $\bar{P}$ is swapping equivalent to $P$.
The componentwise union of $\bar{P}$ and $P_0$ is a process
$\bar{P}_0$ with $\bar{P} \goesto[b] \bar{P}_0$ and $P_0 \leq \bar{P}_0$.
As $P_0 \in^2 \BDify{\bar{P_0}}\setminus \BDify{\bar{P}}$ we have
$\bar{P} \mathbin{\not\swapeq^\infty} \bar{P}_0$.

  ``(2)'':
  Assume $\BDinf{P}$ is maximal, i.e.
  $\nexists P' \mathbin{\swapeq^\infty} P\, \exists Q.
    P' \mathbin\leq Q \wedge P' \mathbin{\not\swapeq^\infty} Q$.
  Then via the contraposition of (1), $\BDify{P}$ is maximal.
  \qed
\end{proof}
  \end{TR}

\begin{lemma}\rm\label{lem-scmaximals}
  Let $N$ be a structural conflict net, and $P$ be a GR-process of $N$.

\begin{TR}
  \vspace{-1.5ex}
  \begin{enumerate}
    \item[(1)] If $\BDify{P}$ is not maximal, then
        $P$ is not maximal, and
    \item[(2)] If $\BDinf{P}$ is weakly maximal,
        then $\BDinf{P}$ is run-maximal.
  \end{enumerate}
\end{TR}
\begin{CONCUR}
    If $\BDinf{P}$ is weakly maximal, then $\BDinf{P}$ is run-maximal.
\end{CONCUR}
\end{lemma}
  \begin{TR}
\begin{proof}
  ``(1)'':
  Take $\R \supsetneq \BDify{P}$. Take a minimal
  $Q_0 \in^2 \R \mathop{\setminus} \BDify{P}$.
  $Q_0$ can be written as $((\SS, \TT, \FF, \MM_0), \pi)$.
  Let $t$ be a maximal element in $\TT$ with respect to $\FF^+$.
  Then $Q_0 \upharpoonright (\TT \setminus \{t\}) =: Q'_0$ is a process
  and $Q'_0 \mathbin{\in^2} \BDify{P}$.
  Hence there exists finite $P'_0, Q'$ such that $Q'_0 \leq Q' \swapeq P'_0 \leq P$.
  Moreover, there are $Q'_1,\ldots,Q'_n\in^2 \BDify{P}$ with $Q_n=Q'$
  and $Q'_{i-1}\goesto[a_i]Q'_i$ for $i=1,\ldots,n$.

  $\pi(t)$ is some transition $b$ of $N$, so $Q'_0 \goesto[b] Q_0$.
  We now show by induction on $i\in\{1,\ldots,n\}$ that there are
  $Q_1,\ldots,Q_n\in^2 \R\mathop{\setminus}\BDify{P} $
  with $Q'_i\goesto[b]Q_i$ and \mbox{$\BD{Q_{i-1}}\goesto[a_i]\BD{Q_i}$} for $i\mathbin=1,\ldots,n$.
  Namely, given $Q_{i-1}$, as $Q_{i-1}\mathbin{\not\in^2}\BDify{P}$ we have
  $Q_{i-1}\not\swapeq^* Q'_i\in^2\BDify{P}$.
  Using that $Q'_{i-1} \goesto[a_i] Q''_i\swapeq^* Q'_i$ and $Q'_{i-1} \goesto[b] Q_{i-1}$,
  this implies $a_i\mathbin{\neq} b$ by \refobs{swaptrans}(a).
  Now \reflem{scdiamond} yields a $Q_i\mathbin{\in^2}\R$ such that 
  $Q'_{i} \goesto[b] Q_{i}$ and $\BD{Q_{i-1}} \goesto[a_i] \BD{Q_i}$.
  As $\BDify{P}$ is prefix closed, we have $Q_i\mathbin{\not\in^2}\BDify{P}$.

  Since $Q'_n\swapeq^* P'_0$ and $Q'_n\goesto[b] Q_{n}$,
  there is a $P_0$ with $P'_0\goesto[b] P_0$ and $P_0\swapeq^*Q_{n}$, using \refobs{swaptrans}(b).
  Hence $P_0\in^2\R\mathop{\setminus}\BDify{P}$.

  Now let $t$ be any transition in $P:=(\NN,\pi_P)$ that is not included in $P'_0$.
  Then there are $P'_1,\ldots,P'_{m+1} \leq P$ with
  $P'_{i}\goesto[c_i]P'_{i+1}$ for $i=0,\ldots,m$ and $c_m=\pi_P(t)$.
  Exactly as above, by induction on $i$, there are
  $P_1,\ldots,P_{m+1}\in^2 \R\mathop{\setminus}\BDify{P} $
  with $P'_{i+1}\goesto[b]P_{i+1}$ and \mbox{$\BD{P_{i}}\goesto[c_i]\BD{P_{i+1}}$} for $i=0,\ldots,m$.
  Moreover, since \mbox{$P'_{m}\goesto[c_i]P'_{m+1}$} and $P'_{m}\goesto[b]P_{m}$,
  we have \plat{$\widehat P'_m \goesto[\{c_m,b\}]$} by \reflem{scdiamond}.
  By \refobs{bisimulation}(d) we furthermore have \plat{$M_0 \goesto[] \widehat
  P'_{m}$}, where $N=:(S,T,F,M_0)$.
  Hence, as $N$ is a structural conflict net, $\precond{b}\cap\precond{c_m}=\emptyset$.

  Since $\widehat P'_0 \subseteq \precond{b}$, by \refobs{bisimulation}(b), and the
  tokens in the preplaces of $b$ cannot be consumed by the $\pi_P$-image of any
  transition of $P$ that fires after $P'_0$ has been executed, $P$ can be extended
  with the transition $b$, and hence is not maximal.

  ``(2)'':
  Assume $\BDinf{P}$ is weakly maximal. Then there is a maximal process
  $P' \in \BDinf{P}$. By (1) if $\BDify{P'}$ would not be maximal, neither
  would $P'$ be.
  Hence $\BDify{P} = \BDify{P'}$ is maximal.
  \qed
\end{proof}
  \end{TR}

\noindent
Finally, we are able to show, using \refthm{onemaximalrun}, that a countable,
binary-conflict-\!- free net has only one maximal BD-process.
In case of a conflict-free structural conflict net we
can do the stronger statement that it has only one weakly maximal
BD-process, i.e.\ only one GR-process up to swapping equivalence.

\begin{lemma}\rm\label{lem-run-conflictfree-max}
  Let $N$ be a binary-conflict-\!-free net.
  \vspace{-1.5ex}

\begin{enumerate}
\item[(1)]
  $N$ has at most one run-maximal BD-process.
\item[(2)]
  If $N$ moreover is countable, then it has exactly one run-maximal BD-process.
\end{enumerate}
\end{lemma}
\begin{proof}
Suppose $N$ had two run-maximal BD-processes $\BDinf{P}$ and $\BDinf{P'}$.
Then $\BDify{P}$ and $\BDify{P'}$ are maximal BD-runs.
By \refthm{onemaximalrun} $N$ has only one maximal BD-run.
Hence $\BDify{P} = \BDify{P'}$ and thus $\BDinf{P} = \BDinf{P'}$.

Now assume that $N$ is countable.
By \refthm{onemaximalrun}, $N$ has a maximal BD-run $\R$.
By \reflem{bdprocessexists} there is a process $P$ with $\BDify{P} = \R$.
By \refdf{bdprocmax} $\BDinf{P}$ is run-maximal, so at least one
run-maximal BD-process exists.
\qed
\end{proof}

\begin{theorem}\rm\label{thm-conflictfree-max}
  Let $N$ be a countable binary-conflict-\!-free net.
  
  $N$ has exactly one maximal BD-process.
\end{theorem}
\begin{proof}
  By Lemmas~\ref{lem-maxorder} and~\ref{lem-cfmaximals} the
  notions of maximality and run-maximality coincide for $N$, and the
  result follows from \reflem{run-conflictfree-max}.
  \qed
\end{proof}
The net of \reffig{uncountable} is an example of an uncountable 
binary-conflict-\!-free net without any maximal or run-maximal BD-process.

\begin{theorem}\rm\label{thm-conflictfree-wm}
  Let $N$ be a conflict-free structural conflict net.

  $N$ has exactly one weakly maximal BD-process, i.e.\ exactly one
  maximal GR-process up to $\swapeq^\infty$.
\end{theorem}
\begin{proof}
  By Lemmas~\ref{lem-maxorder} and~\ref{lem-scmaximals} the
  three maximality notions coincide for $N$, and the ``at most
  one''-direction follows from \reflem{run-conflictfree-max}.

  \begin{TR}By \refobs{bisimulation}(a),\end{TR}\begin{CONCUR}Surely,\end{CONCUR}
  $N$ has at least one process (with an
  empty set of transitions). By Lemma 2 in \cite{glabbeek11ipl}, which
  in turn invokes Zorn's lemma, every GR-process is a prefix of a maximal
  GR-process. Hence $N$ has a maximal GR-process, and thus a maximal
  GR-process up to $\swapeq^\infty$.
  \qed
\end{proof}

\noindent
The assumption that $N$ is a structural conflict net is essential in
\refthm{conflictfree-wm}. The net in \reffig{weakmaximal} is countable
(even finite) and conflict-free, yet has multiple maximal GR-process up to $\swapeq^\infty$.

We can now justify BD-processes as an abstract notion of process for structural conflict
nets since we obtain exactly one maximal abstract process if and only if the underlying
net is conflict-free.

\begin{corollary}\rm\label{cor-structuralconflictbdprocs}
  Let $N$ be a structural conflict net.

  $N$ is conflict-free iff $N$ has exactly one maximal BD-process,
  which is the case iff $N$ has exactly one
  maximal GR-process up to $\swapeq^\infty$.
\end{corollary}
\begin{proof}
    All three notions of maximality coincide for structural conflict nets according to
    \reflem{scmaximals} and \reflem{maxorder}.

    ``$\Rightarrow$'': By \refthm{conflictfree-wm}.

    ``$\Leftarrow$'': By \refthm{structuralconflictbdprocsoneway}.
\qed
\end{proof}

\bibliographystyle{eptcsalpha}
%\bibliography{petri}

\clearpage
\thispagestyle{empty}
\rput(5.5,-9){\includegraphics[width=14cm]{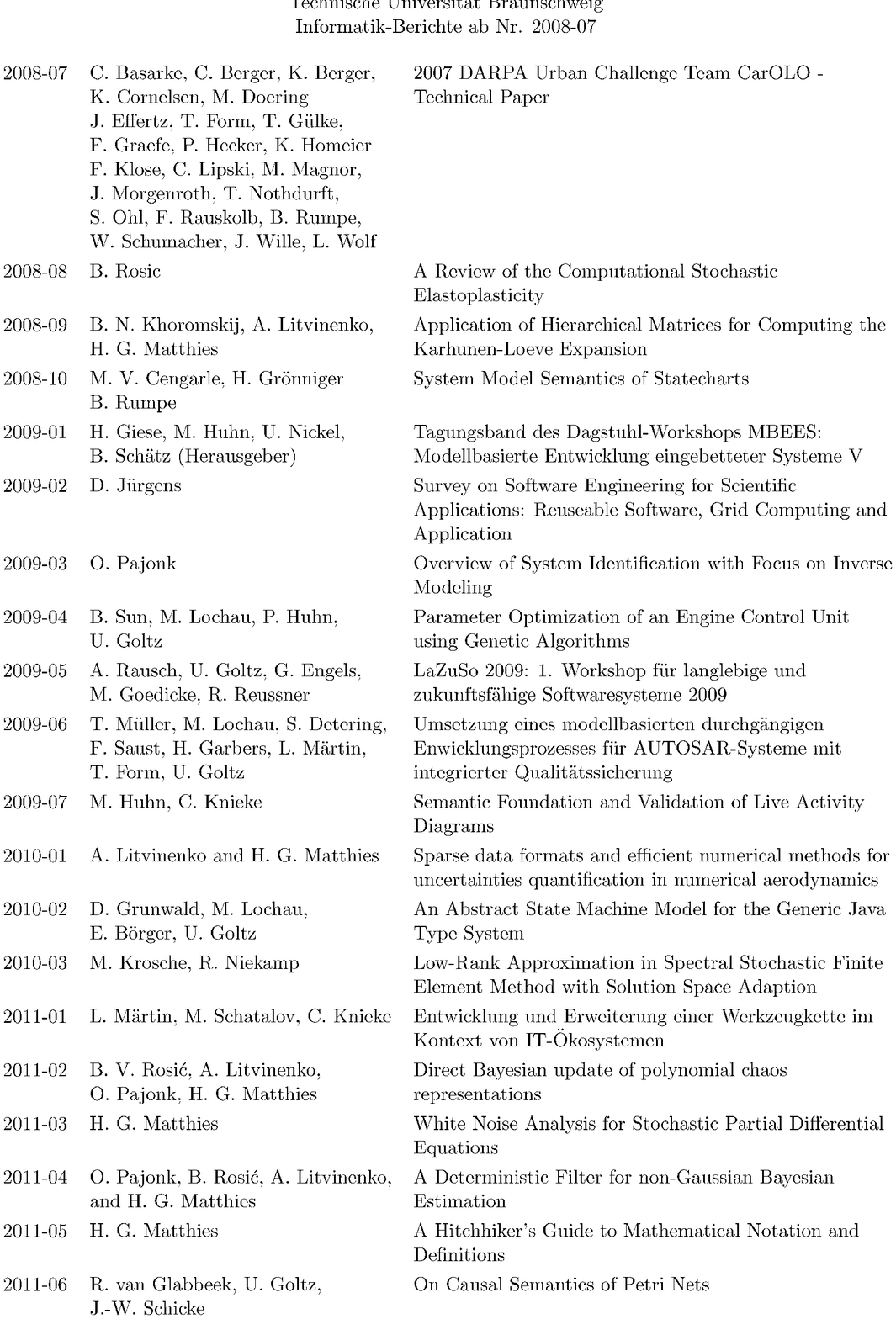}}

\end{document}